\documentclass{article}

\usepackage{PRIMEarxiv}

\usepackage[utf8]{inputenc} 
\usepackage[T1]{fontenc}    
\usepackage{hyperref}       
\usepackage{url}            
\usepackage{booktabs}       
\usepackage{amsfonts}       
\usepackage{nicefrac}       
\usepackage{microtype}      
\usepackage{lipsum}
\usepackage{fancyhdr}       
\usepackage{graphicx}       
\graphicspath{{media/}}     
\usepackage[T1]{fontenc}
\usepackage[utf8]{inputenc}
\usepackage{graphicx}
\usepackage{amsmath}
\usepackage{xspace}
\usepackage{booktabs}
\usepackage{algorithm}
\usepackage{algorithmic}
\usepackage{latexsym}
\usepackage{xargs}
\usepackage{xcolor}
\usepackage{amssymb}
\usepackage{todonotes}
\usepackage{makecell}
\usepackage{amsthm}
\usepackage{mathtools}
\usepackage{etoolbox}
\usepackage{tikz}

\usepackage{dsfont}
\usepackage{soul}
\usepackage{url}
\usepackage{multirow}
\usepackage{cleveref}
\usepackage{amsmath}
\DeclareMathOperator*{\argmax}{arg\,max}

\DeclareFontFamily{U}{mathx}{\hyphenchar\font45}
\DeclareFontShape{U}{mathx}{m}{n}{
      <5> <6> <7> <8> <9> <10>
      <10.95> <12> <14.4> <17.28> <20.74> <24.88>
      mathx10
      }{}
\DeclareSymbolFont{mathx}{U}{mathx}{m}{n}
\DeclareFontSubstitution{U}{mathx}{m}{n}
\DeclareMathAccent{\widebar}{0}{mathx}{"73}
\newcommand{\Oh}{\ensuremath{\mathcal{O}}}

\usepackage{xspace}
\newcommand{\reallocEFone}{\textsc{Good-Bounded EF1-Reformability}\xspace}
\newcommand{\realloc}{\textsc{Good-Bounded $\beta$-Reformability}\xspace}
\newcommand{\reallocEFX}{\textsc{Good-Bounded EFX-Reformability}\xspace}
\newcommand{\reallocEF}{\textsc{Good-Bounded EF-Reformability}\xspace}

\newcommand{\defproblem}[3]{%
  \begin{center}
  \fbox{\begin{minipage}{0.97\linewidth}
    \textbf{#1}\par\smallskip
    \textbf{Input:} #2\par
    \textbf{Question:} #3
  \end{minipage}}%
  \end{center}
}

\usetikzlibrary{arrows.meta}

\newtheorem{theorem}{Theorem}[section]
\newtheorem{lemma}[theorem]{Lemma}

\newtheorem{observation}[theorem]{Observation}

\crefname{theorem}{Theorem}{Theorems}
\Crefname{theorem}{Thm}{Thms}
\crefname{lemma}{Lemma}{Lemmas}
\Crefname{lemma}{Lem}{Lems}

\title{Reaching Fairness by Reallocating Goods}

\author{
Robert Bredereck\textsuperscript{1} \quad
Eva Deltl\textsuperscript{1} \quad
Tanmay Inamdar\textsuperscript{2} \quad
Pallavi Jain\textsuperscript{2} \quad
Pranjal Pandey\textsuperscript{2} \quad
Bin Sun\textsuperscript{1}
\\[0.7em]
\textsuperscript{1} Institute of Computer Science, TU Clausthal,
Clausthal-Zellerfeld, Germany
\\
\textsuperscript{2} Department of Computer Science and Engineering,
IIT Jodhpur, Jodhpur, Rajasthan, India
\\[0.5em]
\texttt{\{robert.bredereck,eva.deltl,bin.sun\}@tu-clausthal.de}
\\
\texttt{taninamdar@gmail.com, \{pallavi,d23csa005\}@iitj.ac.in}
}

\begin{document}
\maketitle

\begin{abstract}
Fair allocation of indivisible goods has largely been studied under the
assumption that no prior allocation exists.  Motivated by practical
settings with pre-existing (and possibly unfair) allocations, 
we study how to achieve fairness through limited
reallocations. Building on recent work on
reformability/reallocations, we consider three fairness notions---envy-freeness (EF),
envy-freeness up to one good (EF1), and envy-freeness up to any good
(EFX)---and optimize 
the number of goods
reallocated. 
We analyze
both the classical and parameterized complexity of these problems,
providing a comprehensive analysis across multiple fairness notions. 

\end{abstract}


\section{Introduction}

Fair allocation of indivisible goods is a well-studied problem in
Economics, Mathematics, and Computer Science~\cite{DBLP:journals/ai/AmanatidisABFLMVW23,DBLP:journals/sigecom/AzizLMW22,brandt2016handbook,guo2023survey,DBLP:journals/jair/LiuLSW24,maskin1987fair}.  
Initiated by the work of
Steinhaus~\cite{steinhaus1948problem}, several notions of fairness have
been proposed, including proportionality, envy-freeness (most popular),
the maximin share, and Nash social welfare~\cite{DBLP:journals/ai/AmanatidisABFLMVW23}.  Most of the prior work
assumes that goods are initially unallocated.  
However, in many practical settings, an allocation is already in place
before fairness is evaluated, and it may fail to satisfy the relevant fairness requirements. This can occur when preferences change or become known only after an initial assignment, when agents or goods are added or removed, or when the fairness criterion itself becomes relevant only after the
allocation has been implemented. Examples include computational resources
assigned to users, students assigned to courses or rooms, and employees
assigned to tasks or projects. In such settings, one may want to restore
or achieve fairness while changing as few assignments as possible.  Moulin also discusses fair division arising from inheritance disputes, divorce settlements, etc., where resources often already have an incumbent allocation and fairness becomes relevant afterward~\cite{DBLP:books/daglib/0017734}. More broadly, these settings have motivated work on control in fair division and on repairing allocations through addition or deletion of goods \cite{aziz2016control,ijcai2025p417,boehmer2024multivariate,DBLP:journals/algorithmica/DornHS21}.

Recently, Chandramouleeswaran et al.~\cite{DBLP:conf/ifaamas/Chandramouleeswaran25}
and Yuen et al.~\cite{DBLP:conf/isaac/Yuen0KS25} independently studied settings in which goods are already allocated---possibly unfairly---and
the goal is to reallocate them according to a predefined rule in order
to achieve fairness. Yuen et al.~\cite{DBLP:conf/isaac/Yuen0KS25} focus
on envy-freeness up to one good (EF1) as the fairness criterion and
consider pairwise exchanges of goods between agents as the reallocation
rule. Chandramouleeswaran et al.~\cite{DBLP:conf/ifaamas/Chandramouleeswaran25}
consider EF1 and EFX, focusing on identical or binary valuations. They
study sequences of good transfers in which no step creates new
EF1-envy among the agents.

\looseness -1
In this paper, we extend their work by considering envy-freeness (EF),
envy-freeness up to one good (EF1), and envy-freeness up to any good
(EFX) as fairness criteria.
Our reallocation objective is to minimize the number of reallocated goods.
From a practical perspective, 
it may be desirable to
\emph{minimize the number of goods} involved in the reallocation, since transferring resources can be
costly and disruptive, e.g.\ moving specialized equipment
between users or locations may incur significant logistical overhead,
so only a small number of changes should be made.
Our objective differs from the framework of Chandramouleeswaran
et al.~\cite{DBLP:conf/ifaamas/Chandramouleeswaran25}, which requires a
sequence of \emph{valid transfers}. We only constrain the final allocation and
the number of goods whose owners change. Although the two models appear closely related, they are fundamentally different: a yes-instance of our
problem may be a no-instance in their model, as illustrated in Observation~\ref{obs:example}.


\paragraph{Formal Model.}
Yuen et al.~\cite{DBLP:conf/isaac/Yuen0KS25} study this setting under the name
{\sc Reformability} while Chandramouleeswaran et al.~\cite{DBLP:conf/ifaamas/Chandramouleeswaran25} study the {$\beta$-\sc Restoration} problem, where $\beta\in \text{\{EF1, EFX\}}$. We use \realloc{} for the variant in which the reallocation cost is the number of goods whose owner changes, where 
$\beta\in \text{\{EF, EF1, EFX\}}$. We define the problem as follows.




\defproblem{\realloc}{A set of
agents $A = \{a_1,\dots,a_n\}$, a set of indivisible goods
$G = \{g_1,\dots,g_m\}$, a valuation function $v_i \colon G \rightarrow
\mathbb{Z}_{\geq 0}$ for each agent $a_i \in A$, an allocation
$\sigma =(\sigma_1,\dots,\sigma_n)$, and an integer $k$.}{Can we reach a
$\beta$-allocation from $\sigma$ by reallocating at most $k$ goods?}

\paragraph{Our Contribution.} Since EF is NP-hard, \reallocEF is also NP-hard. This follows by considering an initial allocation in which all goods are assigned to a single agent and setting $k=m$. Given the computational intractability of EF, we extend our study to the notions of EF1 and EFX. Although EF1 is polynomial-time solvable for monotone valuations~\cite{lipton04}, we establish NP-hardness for both \reallocEFone and \reallocEFX even for two agents [\Cref{thm:nptwo,thm:np}].

We then investigate the complexity of these problems under special valuation functions, in particular, identical and binary. We show that \reallocEFone is polynomial-time solvable for identical valuations when the number of agents is two [\Cref{thm:ef1-two-identical}]. Unfortunately, identical valuations do not lead to tractability beyond two agents: \reallocEFone is weakly NP-hard for identical valuations already when $n=3$ [\Cref{thm:npidentical-n3}], and in fact, strongly NP-hard for an arbitrary number of agents [\Cref{thm:npidentical}]. \reallocEFX,  is already NP-hard for identical valuations when the number of agents is two. Next, we consider binary valuations and show that all variants are polynomial-time solvable when the number of agents is constant [\Cref{thm:fpt-nvmax}], but become NP-hard for an arbitrary number of agents [\Cref{thm:npbin,thm:np}]. Furthermore, when valuations are both identical and binary, all problems admit polynomial-time algorithms [\Cref{thm:polyidx-ef1,thm:polyidx}]. Our complexity results match those of Yuen et al.~\cite{DBLP:conf/isaac/Yuen0KS25} and Chandramouleeswaran et al.~\cite{DBLP:conf/ifaamas/Chandramouleeswaran25} in all cases they consider.

Motivated by these hardness results, even under restricted valuation classes, we subsequently study the problems from a parameterized complexity perspective. 
With respect to the number of goods $m$, we design FPT algorithms for all the variants by showing that $n$ can be bounded by a polynomial function of $m$ [\Cref{cor:fptm}]. Since the problems are para-NP-hard with respect to $n$ (\Cref{thm:nptwo,thm:np}) and $v_{\max}$ (maximum valuation of an agent for a good) (\Cref{thm:npbin,thm:np}), we study the combined parameter $n+v_{\max}$, and design FPT algorithms [\Cref{thm:fpt-nvmax,thm:fpt-ef1-nvmax}]. Furthermore, for all the variants, we also show $W[2]$-hardness with respect to $k$ [\Cref{thm:npbin,thm:npbinX}], while also providing an XP algorithm with respect to $k$ [\Cref{cor:fptm}]. Despite the problems being FPT with respect to $n+v_{\max}$, we additionally consider the parameter $n+k+v_{\max}$ and design FPT algorithms [\Cref{thm:fpt-nvmaxk}] whose running time is singly exponential in the parameters, improving upon the doubly-exponential ILP-based algorithms for $n+v_{\max}$; this is particularly advantageous when $k$ is small relative to $(v_{\max}+1)^n$.

A notable message of our results is that, under a reallocation budget,
the computational gap between EF1 and the stronger notions EF and EFX
largely disappears. In classical fair division, EF1 is efficiently
computable under broad conditions~\cite{lipton04}, while EF is hard
even under restricted valuations~\cite{AZIZ201571}. In our setting,
EF1, EF, and EFX exhibit essentially the same hardness behavior across
valuation classes and parameterizations; the main exception is two
agents with identical valuations, where EF1 remains tractable.




\begin{table}[t]
\centering
\small
\setlength{\tabcolsep}{3pt}
\renewcommand{\arraystretch}{1.18}

\begin{tabular}{@{}llcccc@{}}
\toprule
\textbf{Fairness} & \textbf{$n$}
& \textbf{General}
& \textbf{Identical}
& \textbf{Binary}
&  \textbf{Id. \& Binary} \\
\midrule

\multirow{3}{*}{EF1}
& $2$
& NP-h~[\Cref{thm:nptwo}]
& P~[\Cref{thm:ef1-two-identical}]
& P~[\Cref{thm:fpt-nvmax}]
& P~[\Cref{thm:fpt-nvmax}] \\

& constant
& NP-h~[\Cref{thm:nptwo}]
& NP-h~[\Cref{thm:npidentical-n3}]
& P~[\Cref{thm:fpt-nvmax}]
& P~[\Cref{thm:fpt-nvmax}] \\

& general
& NP-h~[\Cref{thm:npbin}]
& NP-h~[\Cref{thm:npidentical}]
& NP-h~[\Cref{thm:npbin}]
& P~[\Cref{thm:polyidx-ef1}] \\

\midrule

\multirow{3}{*}{EF/EFX}
& $2$
& NP-h~[\Cref{thm:np}]
& NP-h~[\Cref{thm:np}]
& P~[\Cref{thm:fpt-nvmax}]
& P~[\Cref{thm:fpt-nvmax}] \\

& constant
& NP-h~[\Cref{thm:np}]
& NP-h~[\Cref{thm:np}]
& P~[\Cref{thm:fpt-nvmax}]
& P~[\Cref{thm:fpt-nvmax}] \\

& general
& NP-h~[\Cref{thm:np}]
& NP-h~[\Cref{thm:np}]
& NP-h~[\Cref{thm:np}]
& P~[\Cref{thm:polyidx,thm:polyidx-ef1}] \\

\bottomrule
\end{tabular}
\vspace{.5em}
\caption{Complexity landscape of \realloc. 
P denotes polynomial-time solvability, NP-h denotes NP-hardness, and $n$ is the number of agents.}
\label{table:hard}
\end{table}

See \Cref{table:hard} for an overview of classical complexity and
\Cref{fig:params-overview} for parameterized complexity. 

\newcommand{\Wone}{$W[1]$}
\newcommand{\Wtwo}{$W[2]$}

\tikzset{
	textpara/.style = {
		align = center,
		font = \small,
	},
	para/.style = {
		textpara,
		draw,
		rectangle,
		rounded corners = 1.5mm,
		inner sep = 4pt,
		line width = 0.45pt,
		text width = 2.55cm,
		minimum height = 1.05cm
	},
	paranp/.style = {
		para,
		fill=red!12,
		draw=red!55!black,
	},
	paraxp/.style = {
		para,
		fill=orange!16,
		draw=orange!70!black,
	},
	parapk/.style = {
		para,
		fill=green!14,
		draw=green!50!black,
	},
	parapoly/.style = {
		para,
		fill=blue!12,
		draw=blue!55!black,
	},
	arrow/.style = {
		semithick,
		->,
		>=Stealth,
		draw=black!65
	},
	reductionarrow/.style = {
		semithick,
		->,
		>=Stealth,
		dashed,
		draw=black!65
	}
}

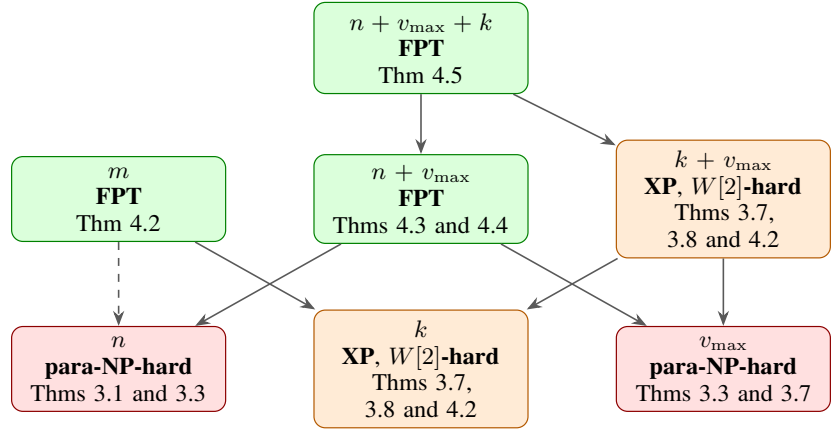
\begin{figure}[t]
	\centering
	\begin{tikzpicture}

		\node[parapk] (nvmaxk) at (0, 3.0)
			{$n+v_{\max}+k$\\
			\textbf{FPT}\\
			\Cref{thm:fpt-nvmaxk}};

		\node[parapk] (m) at (-4.0, 1.0)
			{$m$\\
			\textbf{FPT}\\
			\Cref{cor:fptm}};

		\node[parapk] (nvmax) at (0, 1.0)
			{$n+v_{\max}$\\
			\textbf{FPT}\\
			\Cref{thm:fpt-nvmax,thm:fpt-ef1-nvmax}};

		\node[paraxp] (kvmax) at (4.0, 1.0)
			{$k+v_{\max}$\\
			\textbf{XP}, \textbf{\Wtwo-hard}\\
			\Cref{cor:fptm,thm:npbin,thm:npbinX}};

		\node[paranp] (n) at (-4.0, -1.25)
			{$n$\\
			\textbf{para-NP-hard}\\
			\Cref{thm:nptwo,thm:np}};

		\node[paraxp] (k) at (0, -1.25)
			{$k$\\
			\textbf{XP}, \textbf{\Wtwo-hard}\\
			\Cref{cor:fptm,thm:npbin,thm:npbinX}};

		\node[paranp] (vmax) at (4.0, -1.25)
			{$v_{\max}$\\
			\textbf{para-NP-hard}\\
			\Cref{thm:npbin,thm:np}};

		\draw[arrow] (nvmaxk) -- (nvmax);
		\draw[arrow] (nvmaxk) -- (kvmax);

		\draw[arrow] (nvmax) -- (n);
		\draw[arrow] (nvmax) -- (vmax);

		\draw[arrow] (kvmax) -- (k);
		\draw[arrow] (kvmax) -- (vmax);

		\draw[reductionarrow] (m) --
			(n);

		\draw[arrow] (m) -- (k);

	\end{tikzpicture}

	\caption{
		Overview of our parameterized results. Here, $n$ denotes the number of
		agents, $m$ the number of goods, $k$ the reallocation budget, and
		$v_{\max}$ the maximum value of a good. Solid arrows indicate direct
		parameter implications. The dashed edge $m \to n$ follows from the agent-reduction argument in
		\cref{lem:red-ef1-efx}.
	}
	\label{fig:params-overview}
\end{figure}




\section{Preliminaries}
Let $[n] = \{1,\dots, n\}$. 

\smallskip
\noindent\emph{\bf Model.}
Let $A = \{a_1,\dots,a_n\}$ be a set of agents and $G = \{g_1,\dots,g_m\}$
a set of indivisible goods.
Each agent $a_i \in A$ has an \emph{additive} valuation function $v_i : G \rightarrow \mathbb{Z}_{\ge 0}$, such that for every subset $S\subseteq G$ we have $v_i(S) = \sum_{g \in S} v_i(g)$. We say that valuations are \emph{binary} if for every agent $a_i \in A$ and good $g \in G$, $v_i(g) \in \{0,1\}$. If all agents share the same valuation
function~$v$, such that for all $a_i \in A$,
$v_i = v$, we say that valuations are \emph{identical}.

An allocation is a partition $\sigma = (\sigma_1,\dots,\sigma_n)$ of $G$, where $\sigma_i \subseteq G$ is the bundle assigned to agent $a_i$,
$\sigma_i \cap \sigma_j = \emptyset$ for $i \neq j$, and 
$\bigcup_{a_i \in A}  \sigma_i = G$. 
We denote by $\sigma(g)$ the agent holding good $g$ under allocation $\sigma$. We write $v|_A$ for the restriction of $v$ to $A$ and $\sigma|_A$ for the restriction of $\sigma$ to $A$.
We denote by $v_{\max} = \max_{a_i \in A,\, g \in G} v_i(g)$ the maximum value assigned to any single good. 

\smallskip
\noindent\emph{\bf Fairness.}
Given an allocation $\sigma$, agent $a_i$ envies $a_j$ if
$v_i(\sigma_i) < v_i(\sigma_j)$.
Allocation $\sigma$ is \emph{envy-free (EF)} if no agent envies any
other agent. 

An allocation $\sigma$ is \emph{envy-free up to any good (EFX)}~\cite{DBLP:journals/teco/CaragiannisKMPS19,DBLP:journals/siamdm/PlautR20} if for every pair of agents $a_i,a_j \in A$ with $\sigma_j \neq \emptyset$, and for every good $g \in \sigma_j$, we have $v_i(\sigma_i) \;\ge\; v_i(\sigma_j \setminus \{g\}).$
We say agent $a_i$ \emph{EFX-envies} agent $a_j$ (under $\sigma$) if there exists a good $g \in \sigma_j$ 
such that $v_i(\sigma_i) < v_i(\sigma_j \setminus \{g\})$.

An allocation $\sigma$ is \emph{envy-free up to one good (EF1)} if for every pair of agents $a_i,a_j \in A$ with $\sigma_j \neq \emptyset$, there exists a good $g \in \sigma_j$ such that $v_i(\sigma_i) \;\ge\; v_i(\sigma_j \setminus \{g\}).$
Agent $a_i$ \emph{EF1-envies} agent $a_j$ (under $\sigma$) if $v_i(\sigma_i) < v_i(\sigma_j \setminus \{g\})$ for every $g \in \sigma_j$.

A valuation function $v_i$ is \emph{monotone} if $v_i(S) \leq v_i(T)$ for all $S \subseteq T \subseteq G$. Additive valuations with $v_i : G \rightarrow \mathbb{Z}_{\ge 0}$ are automatically monotone; EF1 allocations can be computed in polynomial time under monotone valuations~\cite{lipton04}.

\smallskip
\noindent\emph{\bf Reallocation Distance.}
Given two allocations $\sigma$ and $\sigma'$, define their distance as
$
d(\sigma,\sigma')
= \bigl|\{\, g \in G \mid \sigma(g) \neq \sigma'(g) \,\}\bigr|,
$
i.e., the number of goods whose assigned agent changes.

\smallskip
\noindent\emph{\bf Parameterized Complexity.}
A \emph{parameterized problem} is a language $L \subseteq \Sigma^{*} \times \mathbb{N}$; the second component is the \emph{parameter}. A problem is \emph{fixed-parameter tractable} (\emph{FPT}) if every instance $(x,\kappa)$ can be solved in time $f(\kappa)\cdot|x|^{O(1)}$ for some computable function $f$. A problem is in \emph{XP} (\emph{slice-wise polynomial}) if it can be solved in time $|x|^{f(\kappa)}$. Note that $\mathrm{FPT} \subseteq \mathrm{XP}$. 
The \emph{$W$-hierarchy} $\mathrm{FPT} \subseteq W[1] \subseteq W[2] \subseteq \cdots$ provides a scale of parameterized intractability: a problem that is \emph{$W[1]$-hard} (resp.\ \emph{$W[2]$-hard}) is not in FPT unless $\mathrm{FPT} = W[1]$.
A problem is \emph{para-NP-hard} with respect to a parameter $\kappa$ if it is NP-hard already for some fixed value of $\kappa$; para-NP-hard problems are not in XP unless $\mathrm{P}=\mathrm{NP}$. We refer to~\cite{DF13} for further background.

\noindent\emph{Integer Linear Programming.}
Our FPT algorithms use the following classical result of Lenstra~\cite{Lenstra83}: feasibility of an integer linear program with $p$ integer variables and encoding size $L$ can be decided in time $p^{O(p)}\cdot\mathrm{poly}(L)$. In particular, when $p$ depends only on the parameter, this yields an FPT algorithm.

\smallskip
\noindent\emph{\bf Relation between fairness criteria.} Interestingly, EF and EFX are related in the world of reallocation. 


\begin{lemma}\label[lemma]{lem:ef-efx}
Given an instance~$(A,G,v,\sigma,k)$ of \reallocEF, one can construct an equivalent instance~$(A,G',v',\sigma',k)$ of \reallocEFX in polynomial time. Furthermore, $|G'|=|G|+|A|(k+1)$ and $v_{max}'=v_{max}$.
\end{lemma}

\begin{proof}
Let~$G=\{g_1,\dots,g_m\}$ and let~$n=|A|$.  We construct~$G'$ by adding $n(k+1)$ dummy goods~$g_{a,i}$, where $a\in A$ and
$i\in\{1,\dots,k+1\}$.  Every dummy good has value zero for every
agent, that is, $v_{a'}(g_{a,i})=0$ for all $a,a'\in A$ and all~$i$.

The allocation~$\sigma'$ coincides with~$\sigma$ on the original goods
and, additionally, assigns the dummy goods~$g_{a,1},\dots,g_{a,k+1}$
to each agent~$a \in A$. 

We show that the two instances are equivalent.

First suppose that the original instance is a yes-instance.  Then there
is an allocation~$\pi$ of~$G$ that is envy-free and differs from~$\sigma$
in at most~$k$ reallocations.  Extend~$\pi$ to an allocation~$\pi'$
of~$G'$ by keeping all dummy goods with their owners from~$\sigma'$.
Since dummy goods have value zero for every agent, the values of all
bundles are unchanged. Thus, for every agent~$a$ and every bundle
$X\subseteq G'$, $v_a(X)=v_a(X\cap G)$.  Hence~$\pi'$ is envy-free
on~$G'$, and therefore also EFX.  Moreover, $\pi'$ differs
from~$\sigma'$ in exactly the same original goods as~$\pi$ differs
from~$\sigma$, so it uses at most~$k'=k$ reallocations.  Thus the
constructed EFX-instance is a yes-instance.

Conversely, suppose that the constructed instance has an EFX
allocation~$\pi'$ differing from~$\sigma'$ in at most~$k'=k$
reallocations.  Since each agent initially owns~$k+1$ dummy goods and
at most~$k$ goods are reallocated in total, every agent still owns at
least one dummy good in~$\pi'$.

We claim that the restriction~$\pi$ of~$\pi'$ to the original
goods~$G$ is envy-free.  Let~$a,a'\in A$.  Since~$a'$ owns at least
one dummy good~$g$ in~$\pi'$, EFX for~$\pi'$ implies
$v_a(\pi'_a) \geq v_a(\pi'_{a'}\setminus\{g\})$.  Because~$g$ has
value zero for~$a$, we have
$v_a(\pi'_{a'}\setminus\{g\}) = v_a(\pi'_{a'})$. Therefore
$v_a(\pi'_a)\geq v_a(\pi'_{a'})$ for all agents~$a,a'$, so~$\pi'$ is
envy-free.  Removing dummy goods does not change any value, hence the
restricted allocation~$\pi$ is envy-free on the original instance.
Since deleting dummy goods cannot increase the number of reallocated
original goods,~$\pi$ differs from~$\sigma$ in at most~$k$ reallocations.

Thus the original instance is a yes-instance if and only if the
constructed instance is a yes-instance.  The construction only
adds~$n\cdot(k+1)$ dummy goods and can thus be computed in polynomial time
in the size of the output. Moreover, the construction preserves the
number of agents and the reallocation budget~$k$.  It also preserves
$v_{\max}$, since all added dummy goods have value zero for every agent. 
\end{proof} 

\paragraph{\bf Comparison to previous models.}
Our work and Chandramouleeswaran et al.~\cite{DBLP:conf/ifaamas/Chandramouleeswaran25} both study how to repair an initially unfair allocation, but measure the distance between two allocations in fundamentally different ways. We use the \emph{Hamming distance} $d(\sigma,\sigma') =
|\{g \in G \mid \sigma(g) \neq \sigma'(g)\}|$ between the initial allocation $\sigma$ and a target allocation $\sigma'$. Thus, \realloc{} asks whether there exists an EF1 allocation within Hamming distance at most $k$ from $\sigma$,
without 
constraints on
intermediate allocations.

By contrast, Chandramouleeswaran et al. require the repair to proceed
through a sequential path of \emph{valid transfers}, where each 
is a single-good move that does not increase EF1-envy, and they minimise
the number of such transfers. The two models measure reallocation cost differently. An EF1 allocation may be close in Hamming distance to the initial allocation even though it is not reachable via valid transfers. Moreover, a valid-transfer path may require many intermediate moves even when its endpoint differs from the initial allocation on only a few goods.

\begin{observation}\label{obs:example}
Our notion of reformability does not require EF1 to be preserved during
the reallocation process. Hence, there exists an instance that can be
made EF1 by reallocating at most $k \geq 2$ goods, even though EF1-Restoration, as defined by Chandramouleeswaran et al.~\cite{DBLP:conf/ifaamas/Chandramouleeswaran25}, is impossible: no valid transfer sequence from the initial allocation reaches an EF1 allocation.
\end{observation}
\begin{proof} Consider the following instance. Let $A=\{a_1,a_2,a_3\}$ be the set of agents and let $G=\{g_1,g_2,g'_1,g'_2,g^\star\}$ be the set of goods. The initial allocation is
\[
\sigma_1=\{g_1,g_2\},\qquad
\sigma_2=\{g'_1,g'_2\},\qquad
\sigma_3=\{g^\star\}.
\]
The agents have additive valuations given by
\[
v_1(x)=
\begin{cases}
1 & \text{if } x\in\{g'_1,g'_2,g^\star\},\\
0 & \text{otherwise},
\end{cases}
\qquad
v_2(x)=
\begin{cases}
1 & \text{if } x\in\{g_1,g_2,g^\star\},\\
0 & \text{otherwise},
\end{cases}
\]
and $v_3(x)=1$ for every $x\in G$.

Let $E(\sigma)$ denote the set of ordered pairs $(a_i,a_j)$ such that
$a_i$ EF1-envies $a_j$ under allocation $\sigma$. Initially, $E(\sigma)=\{(a_1,a_2),(a_2,a_1)\}$.
Indeed, agent $a_1$ values her own bundle at $0$ and agent $a_2$'s bundle at $2$. Removing one good from $\sigma_2$ still leaves value $1$, so $(a_1,a_2)\in E(\sigma)$.
Symmetrically, $(a_2,a_1)\in E(\sigma)$. Agent $a_3$ values her own bundle at $1$ and each of the bundles $\sigma_1,\sigma_2$ at $2$, so after removing one good from either bundle the remaining value is $1$. Hence agent $a_3$ does not EF1-envy anyone.
Moreover, agents $a_1$ and $a_2$ do not EF1-envy agent $a_3$, since $\sigma_3$ contains only one good.

Consider the allocation
\[
\sigma'_1=\{g'_1,g_1\},\qquad
\sigma'_2=\{g'_2,g_2\},\qquad
\sigma'_3=\{g^\star\}.
\]
It is obtained from $\sigma$ by the two reallocations: $g'_1$ is reallocated from $a_2$ to $a_1$ and 
$g_2$ is reallocated from $a_1$ to $a_2$. Furthermore,
$v_1(\sigma'_1)=v_1(\sigma'_2)=v_1(\sigma'_3)=1$
and
$v_2(\sigma'_1)=v_2(\sigma'_2)=v_2(\sigma'_3)=1.$
Thus agents $a_1$ and $a_2$ do not envy anyone. Agent $a_3$ has value $1$ for her own bundle and value $2$ for each of $\sigma'_1$ and $\sigma'_2$; after removing one good from either of those bundles, the remaining value is $1$. Hence $\sigma'$ is EF1.

We now show that two reallocations are necessary. After a single reallocation,
at least one of agents $a_1$ and $a_2$ still receives value $0$ from her own bundle. If agent $a_1$ still receives value $0$, then some other agent holds at least two goods from $\{g'_1,g'_2,g^\star\}$, and therefore agent $a_1$ EF1-envies that agent. The case of agent $a_2$ is symmetric. Hence no allocation reachable from $\sigma$ by one
reallocation is EF1. Therefore the minimum number of reallocations needed to
reach an EF1 allocation is exactly $2$.

It remains to show that no first transfer from $\sigma$ is valid. Let
$\sigma^\star$ be the allocation obtained after an arbitrary single-good
reallocation from $\sigma$. We distinguish cases according to the
reallocated good.

Since $g'_1$ and $g'_2$ have identical values for all agents, it suffices to
consider a representative good $g'$. Similarly, we write $g$ for either
$g_1$ or $g_2$. Suppose that a good is moved from agent $a_1$ to agent $a_3$.
Then the transferred good is a $g$-good, so $\sigma^\star_3=\{g,g^\star\}$, while
$v_2(\sigma^\star_2)=0$
 and 
$v_2(\sigma^\star_3)=2$.
Removing either good from $\sigma^\star_3$ leaves value $1$ for agent $a_2$. Hence $(a_2,a_3)\in E(\sigma^\star)$.
Since $(a_2,a_3)\notin E(\sigma)$, this is not a valid transfer.

If a good is moved from agent $a_2$ to agent $a_3$, then agent $a_3$ receives a $g'$ good. Thus $\sigma^\star_3=\{g',g^\star\}$, while
$v_1(\sigma^\star_1)=0$
 and 
$v_1(\sigma^\star_3)=2$.
Removing either good from $\sigma^\star_3$ leaves value $1$ for agent $a_1$. Hence $(a_1,a_3)\in E(\sigma^\star)$.
Since $(a_1,a_3)\notin E(\sigma)$, this is not a valid transfer.

If a good is moved from agent $a_1$ to agent $a_2$, then agent $a_2$ has a bundle of size $3$. Since agent $a_3$ values every good at $1$, $v_3(\sigma^\star_3)=1$ and $v_3(\sigma^\star_2)=3$. 
After removing any single good from $\sigma^\star_2$, agent $a_3$ still values the remaining bundle at $2$. Hence
$(a_3,a_2)\in E(\sigma^\star)$.
Since $(a_3,a_2)\notin E(\sigma)$, this is not a valid transfer.

If a good is moved from agent $a_2$ to agent $a_1$, then agent $a_1$ has a bundle of size $3$. Since agent $a_3$ values every good at $1$, 
$v_3(\sigma^\star_3)=1$ and
$v_3(\sigma^\star_1)=3$.
After removing any single good from $\sigma^\star_1$, agent $a_3$ still values the remaining bundle at $2$. Hence
$(a_3,a_1)\in E(\sigma^\star)$.
Since $(a_3,a_1)\notin E(\sigma)$, this is not a valid transfer.

Finally, if the good $g^\star$ is moved from agent $a_3$ to agent $a_1$, then
agent $a_3$'s bundle becomes empty while agent $a_1$ has three goods. Hence
$v_3(\sigma^\star_3)=0$ and
$v_3(\sigma^\star_1)=3$.
After removing any single good from $\sigma^\star_1$, agent $a_3$ still values the remaining bundle at $2$, so $(a_3,a_1)\in E(\sigma^\star)$. This pair was not in $E(\sigma)$.

Similarly, if $g^\star$ is moved from agent $a_3$ to agent $a_2$, then $(a_3,a_2)\in E(\sigma^\star)$, again creating new EF1-envy.

Thus, every possible single-good reallocation from $\sigma$ creates at
least one EF1-envy pair that was not present initially. Consequently,
although an EF1 allocation can be reached optimally with two
reallocations, no corresponding transfer path is valid, since validity
requires that no new EF1-envy pair be introduced at any intermediate
step.
\end{proof}

\section{Classical Complexity}
We study \realloc{} under increasingly restricted valuations: general, identical, binary, and finally both.
\subsection{General valuations}
We start with arbitrary additive valuations and show hardness for two agents.

\begin{theorem}\label{thm:nptwo}
\reallocEFone\ is weakly NP-hard even with two agents.
\end{theorem}

\begin{proof}
We reduce from \textsc{Equal cardinality partition}
\cite{DBLP:books/fm/GareyJ79}: Given a set of integers
$S = \{s_1, \dots, s_m\}$, such that $\sum_{i\in [m]} s_i = 2B$, can we partition
$S$ into $S_1, S_2$ such that they sum up to the same value and
$|S_1|=|S_2|$?

We construct an instance $I'=(A,G,v,\sigma,k)$ of \reallocEFone{} as follows.
Let $A = \{a_1,a_2\}$ denote the set of agents. We may assume w.l.o.g. that the largest integer $s_{\max}\le B$ and $m>2$, since otherwise the instance is trivial. We create the set of goods $G$ by adding a good $g_i$ with
$v_1(g_i)=v_2(g_i)=s_i+mB$ for each integer $s_i \in S$, where we
denote by $s_{\max} = \max \{s_i \in S\}$.  We further add a good
$g_{\max}$ and a set $G^\star$ of $\frac{m^2}{2}+1$ goods with
valuations
\begin{align*}
    &v_1(g_{\max})=v_2(g_{\max})=(\frac{m^2}{2}+1)B \\
&v_1(g)=B+\frac{s_{\max}+mB}{\frac{m^2}{2}+1} \quad \text{ and } \quad 
    v_2(g)=0 \quad \text{ for all } g\in G^\star
\end{align*}
We set the initial allocation as follows:
$\sigma_1 = \{g_{\max}\}$ and
$\sigma_2 = \{g_i \mid s_i \in S\} \cup G^\star$.
Finally set $k = \frac{m}{2}$.  The construction can be performed in
polynomial time. 

We next show the equivalence of the two instances.

$(\Rightarrow)$ Let $S_1, S_2$ be an equal-cardinality partition of
$S$.  We define an allocation $\sigma'$ as follows.  Assume
w.l.o.g.\ that $s_{\max} \in S_1$. $\sigma_1' = \sigma_1 \cup \{g_i \mid s_i \in S_2\}$ and
$\sigma_2' = \sigma_2 \setminus \{g_i \mid s_i \in S_2\}$.  As
$|S_1| = |S_2| = \frac{m}{2} = k$ the allocation $\sigma'$ can
be reached by exactly $k$ reallocations.  It remains to show that
$\sigma'$ is an EF1 allocation.
 We have
$v_1(\sigma_1') = (m^2+2)B$ and
$v_1(\sigma_2') = (m^2+2)B + (s_{\max}+mB)$. By definition (and choice of $S_1$) the most valuable good in $\sigma_2'$ according to agent $a_1$ has a value of at least $s_{\max} + mB$, so there exists $g \in \sigma_2'$ such
that $v_1(\sigma_2'\setminus\{g\}) = (m^2+2)B \leq v_1(\sigma_1')$.  Hence
$a_1$ does not envy $a_2$ up to one good.  Furthermore,
$v_2(\sigma_1') = (m^2+2)B$ and
$v_2(\sigma_2') = (\frac{m^2}{2}+1)B$.  As $v_2(g_{\max}) = (\frac{m^2}{2}+1)B$ there exists $g \in \sigma_1'$ such
that $v_2(\sigma_1'\setminus\{g\}) = (\frac{m^2}{2}+1)B \leq v_2(\sigma_2')$.
Thus $a_2$ does not envy $a_1$ up to one good.  We can conclude that
$\sigma'$ is an EF1 allocation reachable by $k$ reallocations.

$(\Leftarrow)$ Let $\sigma'$ be an EF1 allocation reachable from
$\sigma$ by $k$ reallocations.  Note that, initially, in
$\sigma$ the agents valued the sets as follows:
$v_1(\sigma_1) = (\frac{m^2}{2}+1)B$ and
$v_1(\sigma_2) = 3(\frac{m^2}{2}+1)B+(s_{\max}+mB)$, while
$v_2(\sigma_1) = (\frac{m^2}{2}+1)B$ and
$v_2(\sigma_2) = (m^2+2)B$.  By definition the largest good in $\sigma_2$
according to agent $a_1$ satisfies $s_{\max} \le (m+1)B$.  Hence,
for any subset of $\frac{m}{2}-1$ goods from $\sigma_2'$, its value is
at most
\[
\left(\frac{m}{2}-1\right)(m+1)B
= \left(\frac{m^2}{2}-\frac{m}{2}-1\right)B
< \left(\frac{m^2}{2}+1\right)B.
\]
Therefore, any bundle consisting of only $\frac{m}{2}-1$ goods from
$\sigma_2$ has strictly smaller value than $\left(\frac{m^2}{2}+1\right)B$
for agent $a_1$.  As $v_1(\sigma_2) = 3(\frac{m^2}{2}+1)B + s_{\max}$,
if we only reallocate $k-1$ goods from $\sigma_2$ to $\sigma_1$, we have
\begin{align*}
v_1(\sigma_2') &> 3\left(\frac{m^2}{2}+1\right)B + (s_{\max}+mB) -
\left(\frac{m^2}{2}+1\right)B \\
&= 2\left(\frac{m^2}{2}+1\right)B + (s_{\max}+mB)
\end{align*}
and
\[
v_1(\sigma_1') < 2\left(\frac{m^2}{2}+1\right)B,
\]
which, since the maximum good in $\sigma_2'$ has a value of at most $s_{\max} + mB$, contradicts $v_1(\sigma_1') \ge v_1(\sigma_2'\setminus\{g\})$.  It follows that all goods must be reallocated from $\sigma_2$ to $\sigma_1$.

Thus there exists a set $S'$ of $k$ goods such that
$\sigma_1' = \sigma_1 \cup S'$ and $\sigma_2' = \sigma_2 \setminus S'$, with
$v_1(S') \ge \left(\frac{m^2}{2}+1\right)B$. Since $B \geq s_{\max}$  and $m>2$, we get $$v(g)
= B+\frac{s_{\max}+mB}{\frac{m^2}{2}+1}
\le B+\frac{(m+1)B}{\frac{m^2}{2}+1}
< 2B$$ for all $g \in G^\star$. Thus, $\frac{m}{2}\cdot v(g) < mB$, so all goods from $S'$ have to be
partition goods.  Note that all partition goods are equally valued by
the two agents, thus
$v_1(S') = v_2(S') \ge (\frac{m^2}{2}+1)B$.  Moreover,
\[
v_2(\sigma'_2) = (m^2+2)B - v_2(S') \le \left(\frac{m^2}{2}+1\right)B
\]
and
$v_2(\sigma'_1) = \left(\frac{m^2}{2}+1\right)B + v_2(S') \ge (m^2+2)B$.
As $\mathcal{\sigma'}$ is EF1, it must hold that
$v_2(\sigma'_2) \ge v_2(\sigma'_1\setminus \{g\})$ for some $g \in \sigma'_1$.
By construction $g_{\max}$ is the highest-valued good in $\sigma'_1$ and
$v_2(g_{\max}) = (\frac{m^2}{2}+1)B$.  Thus
\[
v_2(\sigma'_2) \ge v_2(\sigma'_1\setminus \{g_{\max}\}) = v_2(S') \ge
\left(\frac{m^2}{2}+1\right)B.
\]
From the two inequalities we get $v_2(\sigma'_2) = (\frac{m^2}{2}+1)B$.

Now let $S_1 = \{s_i \mid g_i \in S'\}$ and $S_2 = \{s_i \mid g_i \in S\setminus S'\}$.  By construction $v_2(g_i) = s_i + mB$, hence
\[
\sum S_1 = \sum S_2 = \left(\frac{m^2}{2}+1\right)B -
\frac{m^2}{2}B = B.
\]
Furthermore, as $S'$ is the reallocated set, $|S'| = k = \frac{m}{2}$,
so $|S_1| = \frac{m}{2}$, and since $S_2 = S \setminus S_1$, also
$|S_2| = \frac{m}{2}$.  This implies that $S_1, S_2$ form an
equal-cardinality partition of $S$.
\end{proof}

\begin{lemma}\label[lemma]{lem:ef-reduction}
There is a polynomial-time many-one reduction from
\textsc{Envy-free allocation} to \reallocEF{} which preserves the
valuation functions and the number of agents.
\end{lemma}

\begin{proof}
Given an instance $(G,A,v)$ with $G=\{g_1,\dots,g_m\}$, construct an
instance $I'=(A,G,v,\sigma,k)$ of \reallocEF{} with the same agents,
goods, and valuations. Set the initial allocation to $\sigma_1=G$ and
$\sigma_i=\emptyset$ for all $i\neq 1$, and set $k=m$.

Since there are only $m$ goods, any allocation of $G$ can be reached
from this initial allocation using at most $m$ reallocations. Hence,
$I'$ admits an envy-free allocation reachable within $k$ reallocations
if and only if the original instance admits an envy-free allocation.

The construction is polynomial and preserves both binary valuations and
the number of agents. Moreover, the same reduction preserves identical
additive valuations and the parameter $n$.
\end{proof}

\begin{theorem}\label{thm:np}
\realloc is NP-hard, even when all agents have binary valuations, for
$\beta\in \{EF,EFX\}$. Moreover, it is weakly NP-hard already for $n=2$ with identical additive valuations;
hence it is para-NP-hard parameterized by $n$. It is also $W[1]$-hard parameterized by the number $n$ of agents, even
when all agents have identical additive valuations encoded in unary.
\end{theorem}
\begin{proof}
By \cref{lem:ef-reduction}, there is a polynomial-time many-one
reduction from \textsc{Envy-free allocation} to \reallocEF{} preserving
binary valuations, identical additive valuations, and the number of
agents. Since \textsc{Envy-free allocation} is NP-hard even under binary
valuations~\cite[Thm.\ 11]{AZIZ201571}, \reallocEF{}
is NP-hard under binary valuations.

Moreover, \textsc{Envy-free allocation} is weakly NP-hard already for $n=2$: with two identical
additive agents, a complete allocation is envy-free iff the goods can be
split into two bundles of equal total value, which is exactly
\textsc{Partition}. Hence \reallocEF{} is weakly NP-hard already for
$n=2$, and thus para-NP-hard parameterized by $n$.

Finally, since \textsc{Envy-free allocation} is $W[1]$-hard with
respect to $n$, as follows from reductions from Unary Bin
Packing~\cite{BBN16}, \reallocEF{} is $W[1]$-hard parameterized by $n$
even for identical additive valuations.

\noindent
By \cref{lem:ef-efx} the same holds for \reallocEFX.
\qedhere
\end{proof}

Thus, for EF and EFX, hardness already holds for two agents with
identical valuations.  For EF1, the situation is different: although
\reallocEFone{} is hard for two agents with general valuations by
\cref{thm:nptwo}, the identical two-agent case is tractable.  We prove
this next, before showing that hardness returns for three identical
agents.

\subsection{Identical valuations}

\smallskip
In contrast to \cref{thm:nptwo}, when the two agents have identical valuations \reallocEFone{} is tractable.

\begin{theorem}\label{thm:ef1-two-identical}
\reallocEFone{} for two agents with identical additive valuations is solvable in $\Oh(m\log m)$ time.
\end{theorem}
\begin{proof}
We show that Algorithm~\ref{alg:ef1-two-agents} computes an EF1 allocation
reachable from $\sigma$ within $k$ reallocations, or correctly reports that no such allocation exists.


\begin{algorithm}[t]
\caption{EF1-Reallocation for Two Agents with Identical Valuations}
\label{alg:ef1-two-agents}
\begin{algorithmic}[1]
\REQUIRE Identical additive valuation $v$, initial allocation
  $\sigma=(\sigma_1,\sigma_2)$, budget $k$.
\ENSURE An EF1 allocation reachable from $\sigma$ by at most $k$
  reallocations, or \textsc{None}.
\STATE $\ell \leftarrow 0$
\WHILE{$(\sigma_1,\sigma_2)$ is not EF1 \textbf{and} $\ell < k$}
  \STATE $i \leftarrow \arg\max_{j\in\{1,2\}} v(\sigma_j)$
  \STATE $g^* \leftarrow \arg\max_{g\in\sigma_i} v(g)$
  \STATE $\sigma_i \leftarrow \sigma_i\setminus\{g^*\}$
  \STATE $\sigma_{3-i} \leftarrow \sigma_{3-i}\cup\{g^*\}$
  \STATE $\ell \leftarrow \ell+1$
\ENDWHILE
\IF{$(\sigma_1,\sigma_2)$ is EF1}
  \RETURN $(\sigma_1,\sigma_2)$
\ELSE
  \RETURN \textsc{None}
\ENDIF
\end{algorithmic}
\end{algorithm}

Let $A=\{a_1,a_2\}$ and $\sigma$ be the initial allocation. If $\sigma$ is EF1, then Algorithm~\ref{alg:ef1-two-agents} returns it immediately. Hence, correctness is immediate in this case.

Now, assume that the current allocation is not EF1. By symmetry, it
suffices to consider the case $v(\sigma_1) > v(\sigma_2)$. Such an allocation is EF1 if and only if
\begin{equation}\label{eq:ef1-two}
  v(\sigma_1)-v(\sigma_2)
  \leq
  \max_{g\in\sigma_1} v(g).
\end{equation} Let $\Delta = v(\sigma_1)-v(\sigma_2)$.
Then $\Delta>\max_{g\in\sigma_1}v(g)$.

First, moving a good $g\in\sigma_2$ to $\sigma_1$ changes the gap from $\Delta$ to $\Delta+2v(g)$, so it cannot help satisfy condition~\eqref{eq:ef1-two}. Hence, while agent $a_1$ has weakly larger value, we only need to consider reallocations from $\sigma_1$ to $\sigma_2$.

Let $g^*\in\arg\max_{g\in\sigma_1}v(g)$. Moving a good $g\in\sigma_1$ to $\sigma_2$ changes the gap to $\Delta-2v(g)$.
Thus, among all such reallocations, moving $g^*$ minimizes the resulting gap.
If $\Delta-2v(g^*)<0$, then after moving $g^*$ agent $a_2$ has larger value. However, the new gap is $v(\sigma_2\cup\{g^*\}) - v(\sigma_1\setminus\{g^*\})
  =
  2v(g^*)-\Delta
  <
  v(g^*)$,
where the strict inequality follows from
$\Delta>\max_{g\in\sigma_1}v(g)=v(g^*)$. Hence the resulting allocation is EF1.
Otherwise, $\Delta-2v(g^*)\geq 0$, so agent $a_1$ still has weakly larger value, and the greedy reallocation leaves the smallest possible remaining gap among all one-good transfers from $\sigma_1$ to $\sigma_2$.

Now consider any shortest sequence of reallocations that reaches an EF1
allocation. We claim that its first step can be chosen to be the greedy step.
Let the goods of $\sigma_1$ be ordered by non-increasing value, $v(g_1)\ge \cdots \ge v(g_{m'})$,
so that the greedy step moves $g_1$. More generally, after $t$ greedy steps,
provided EF1 has not yet been reached, the goods moved from $\sigma_1$ to
$\sigma_2$ are precisely $g_1,\dots,g_t$.
We use the following exchange observation. For every set $T\subsetneq \sigma_1$ with $|T|=t$, we have
\begin{equation}\label{eq:greedy-exchange}
  2v(T)+\max_{g\in \sigma_1\setminus T}v(g)
  \le
  2\sum_{q=1}^t v(g_q)+v(g_{t+1}).
\end{equation}
Indeed, if the maximum good left in $\sigma_1\setminus T$ is some
$g_j$ with $j>t$, then $v(T)\le \sum_{q=1}^t v(g_q)$ and $v(g_j)\le v(g_{t+1})$. If $j\le t$, then $T$ did not include $g_j$, so the largest possible value of $T$ is obtained by taking the first $t+1$
goods except $g_j$; hence
\[
  2v(T)+v(g_j)
  \le
  2\sum_{q=1}^{t+1}v(g_q)-v(g_j)
  \le
  2\sum_{q=1}^t v(g_q)+v(g_{t+1}),
\]
because $v(g_j)\ge v(g_{t+1})$. This proves
\eqref{eq:greedy-exchange}.

Suppose some sequence reaches EF1 after $t$ reallocations. Until the final
step, the valuation of $\sigma_1$ must remain weakly higher; otherwise, as shown above, the allocation would already be EF1 at the first step where $\sigma_2$ becomes more valuable.
Thus the 
sequence moves a set
$T\subseteq \sigma_1$ of $t$ goods from $\sigma_1$ to $\sigma_2$.
If this 
makes $\sigma_2$ more valuable in the final step, then $\Delta - 2v(T) < 0$. Since the greedy sequence moves goods of total value at least $v(T)$ in its first $t$ steps, the greedy sequence also makes $\sigma_2$ more valuable by step $t$. At the first greedy step where $\sigma_2$ becomes more valuable, the allocation is EF1 by the argument above. Hence greedy reaches EF1 within $t$ steps.

Otherwise, after moving $T$, $\sigma_1$ is still more valuable. Since the
resulting allocation is EF1, we have $ \Delta
  \le
  2v(T)+\max_{g\in \sigma_1\setminus T}v(g)$.
By \eqref{eq:greedy-exchange}, $\Delta
  \le
  2\sum_{q=1}^t v(g_q)+v(g_{t+1})$.
But after $t$ greedy steps, the remaining gap is $\Delta - 2\sum_{q=1}^t v(g_q)$, and the largest remaining good in $\sigma_1$ is $g_{t+1}$. Therefore the
greedy allocation after $t$ steps satisfies $\Delta - 2\sum_{q=1}^t v(g_q)
  \le
  v(g_{t+1})$, and is EF1.

Hence, whenever an EF1 allocation is reachable in $t$ reallocations, the
greedy sequence also reaches EF1 within $t$ reallocations. Thus,
Algorithm~\ref{alg:ef1-two-agents} uses the minimum number of reallocations,
returning an EF1 allocation if one is reachable within $k$ reallocations
and \textsc{None} otherwise.


\smallskip\noindent\emph{Running time.}
Sort the goods by decreasing value and maintain the bundles as
max-heaps with their current total values. Each iteration does one
heap extraction, one heap insertion, and constant-time updates of the
bundle values. The EF1 condition is checked in $\Oh(1)$ time from
the bundle values and heap maxima. Since at most $m$ goods are
reallocated, the total running time is $\Oh(m\log m)$.
\qedhere
\end{proof}
However for three identical agents, the problem becomes intractable.

\begin{theorem}\label{thm:npidentical-n3}
\reallocEFone{} is weakly \textup{NP}-hard even when all agents have identical additive valuations and $n = 3$.
\end{theorem}

\begin{proof} We again reduce from \textsc{Equal Cardinality Partition}
\cite{DBLP:books/fm/GareyJ79}. The construction is inspired by the proof of
\cref{thm:nptwo}, but slightly modified.

We construct an instance $I'=(A,G,v,\sigma,k)$ of \reallocEFone{} as follows. We may assume w.l.o.g. that the largest integer $s_{\max}\le B$ and $m>3$, since otherwise the instance is trivial. 
Let $A = \{s,a_1,a_2\}$ denote the set of agents.  We create the set of
goods $G$ by adding a good $x_i$ with
$v(x_i)=s_i+mB$ for each integer $s_i \in S$, where we denote by $s_{\max} = \max \{s_i \in S\}$. We further add a set $G^\star$ of $\frac{m^2}{2}+2$ goods with valuations 
$v(g)=B$ for all  $g\in G^\star$.
We set the initial allocation as follows:
$\sigma_s = G$ and
$\sigma_i = \emptyset$ for $i \in \{1,2\}$.
Finally set $k = m$. The construction can be performed in polynomial time. 

We next show the equivalence of the two instances.

$(\Rightarrow)$ Let $S_1,S_2$ be an equal-cardinality partition of
$S$. Let $X_1=\{x_i\mid s_i\in S_1\}$ and
$X_2=\{x_i\mid s_i\in S_2\}$. We define an allocation $\sigma'$ by 
$\sigma_s'=G^\star$, $\sigma_1'=X_1$, and $\sigma_2'=X_2$.
Since all $m$ partition goods are reallocated and no good of $G^\star$ is
moved, $\sigma'$ can be reached by exactly $k=m$ reallocations.
Because $|S_1|=|S_2|=\frac{m}{2}$ and both sets have total value $B$, we have
\[
v(\sigma_1')=v(\sigma_2')
=\frac{m}{2}\cdot mB+B
=\left(\frac{m^2}{2}+1\right)B .
\]
Moreover,
\[
v(\sigma_s')=|G^\star|B
=\left(\frac{m^2}{2}+2\right)B .
\]
Thus agents $a_1$ and $a_2$ do not envy each other. Agent $s$ does not envy
either $a_1$ or $a_2$, since $v(\sigma_s')>v(\sigma_i')$ for
$i\in\{1,2\}$. Finally, each agent $a_i$ may envy $s$, but removing any good
$g\in G^\star=\sigma_s'$ gives $v(\sigma_s'\setminus\{g\})
=\left(\frac{m^2}{2}+1\right)B
=v(\sigma_i')$.
Hence $\sigma'$ is EF1.

$(\Leftarrow)$ Let $\sigma'$ be an EF1 allocation with
$d(\sigma,\sigma')\le k=m$. We first show that all partition goods
$x_1,\dots,x_m$ must be reallocated away from agent $s$.

Let $M=|G^\star|=\frac{m^2}{2}+2$. Since at most $m$ goods are reallocated
and initially all goods belong to $s$, the two agents $a_1,a_2$ together
receive at most $m$ goods.

Suppose, for contradiction, that some partition good remains with $s$.
Then, even in the case minimizing $s$'s value, agent $s$ keeps at least one
partition good and at least $M-1$ goods from $G^\star$. Hence
\[
v(\sigma_s')
\ge (M-1)B+mB
=
\left(\frac{m^2}{2}+m+1\right)B .
\]
On the other hand, every good has value at most 
$mB+s_{\max}\le (m+1)B$,
because $\sum_i s_i=2B$. Since $\sigma'$ is EF1, for each
$i\in\{1,2\}$ there exists some good $g_i\in\sigma_s'$ such that
\[
v(\sigma_i')\ge v(\sigma_s'\setminus\{g_i\})
\ge v(\sigma_s')-(m+1)B .
\]
Therefore
\[
v(\sigma_1')+v(\sigma_2')
\ge 2v(\sigma_s')-2(m+1)B .
\]
Since the total value of all goods is
\[
v(G)
=
\sum_{i=1}^m(s_i+mB)+MB
=
2B+m^2B+\left(\frac{m^2}{2}+2\right)B
=
\left(\frac{3m^2}{2}+4\right)B,
\]
we also have
\[
v(\sigma_1')+v(\sigma_2')=v(G)-v(\sigma_s').
\]
Combining the two inequalities gives
\[
v(G)-v(\sigma_s')
\ge 2v(\sigma_s')-2(m+1)B,
\]
and therefore
\[
v(\sigma_s')
\le
\frac{v(G)+2(m+1)B}{3}
=
\left(\frac{m^2}{2}+\frac{2m}{3}+2\right)B .
\]
But this contradicts
\[
v(\sigma_s')
\ge
\left(\frac{m^2}{2}+m+1\right)B
\]
for every $m>3$.
Thus all partition goods must be reallocated away from $s$.

Since there are exactly $m$ partition goods and at most $m$ reallocations are
allowed, it follows that precisely the goods $x_1,\dots,x_m$ are reallocated,
and no good of $G^\star$ is moved. Hence
$\sigma_s'=G^\star$ and 
$v(\sigma_s')=\left(\frac{m^2}{2}+2\right)B$.

Now consider any agent $a_i$, where $i\in\{1,2\}$. Since $\sigma'$ is EF1,
agent $a_i$ must not envy $s$ up to one good. All goods held by $s$ have value
$B$, so
\[
v(\sigma_i')
\ge
v(\sigma_s')-B
=
\left(\frac{m^2}{2}+1\right)B .
\]
Let $S_i\subseteq S$ be the set of integers corresponding to the partition
goods assigned to $a_i$. Then
$v(\sigma_i')
=
|S_i|\,mB+\sum_{s_j\in S_i}s_j $.
Because $\sum_{j=1}^m s_j=2B$, we have
$\sum_{s_j\in S_i}s_j\le 2B$ .
If $|S_i|\le \frac m2-1$, then
\[
v(\sigma_i')
\le
\left(\frac m2-1\right)mB+2B
=
\left(\frac{m^2}{2}-m+2\right)B
<
\left(\frac{m^2}{2}+1\right)B,
\]
a contradiction. Hence each agent receives at least $\frac{m}{2}$ partition goods.
Since there are exactly $m$ partition goods in total, both agents receive
exactly $\frac{m}{2}$ of them.

Finally, for each $i\in\{1,2\}$ we have
\[
v(\sigma_i')
=
\frac m2\cdot mB+\sum_{s_j\in S_i}s_j
\ge
\left(\frac{m^2}{2}+1\right)B,
\]
which implies
$\sum_{s_j\in S_i}s_j\ge B $.
Since the two sets $S_1,S_2$ partition $S$ and the total sum is $2B$, we get
$\sum_{s_j\in S_1}s_j=\sum_{s_j\in S_2}s_j=B $.
Thus $S_1,S_2$ form an equal-cardinality partition of $S$.
\end{proof}

\smallskip
Finally, we show that \reallocEFone{} is strongly \textup{NP}-complete even for identical valuations. For this we reduce from \textsc{3-Partition}. Here, we are given $3r$ positive integers $a_1,\dots,a_{3r}$ with $\sum_{j=1}^{3r} a_j = rB$ and $B/4 < a_j < B/2$ for all $j$; the question is whether
$\{a_1,\dots,a_{3r}\}$ can be partitioned into $r$ triples each summing to $B$. 

\begin{theorem}\label{thm:npidentical}
\reallocEFone{} is \textup{NP}-complete even when all agents have
identical additive valuations.
\end{theorem}

\begin{proof}
Given an instance $(r, B, a_1,\dots,a_{3r})$ of \textsc{3-Partition}, we construct an instance
$(\sigma, A, G, k)$ of \reallocEFone{} as follows. Set $\lambda = 3r+1$. Let
$A = \{s, b_1, \dots, b_r\}$ be the set of agents, where $s$ is a
source agent and $b_1, \dots, b_r$ are target agents. The good set
$G$ consists of $3r$ partition goods $x_1, \dots, x_{3r}$ with
$v(x_j) = \lambda a_j$, and $\lambda B+1$ helper goods
$h_1, \dots, h_{\lambda B+1}$, each of value $1$. All agents have the same
valuation $v$. We define the initial allocation $\sigma$ by
$\sigma_s = G$ and $\sigma_{b_i} = \emptyset$ for all $i \in [r]$, and set
the budget to $k = 3r$.  Since \textsc{3-Partition} is strongly \textup{NP}-complete, the numbers $a_j$ and $B$ may be assumed polynomially bounded in the input size, making the number of helper goods $\lambda B+1$ polynomial and the construction polynomial. 

We claim that the \textsc{3-Partition} instance is a YES-instance if and only if the \reallocEFone{} instance has an EF1 allocation reachable within $k$ reallocations.

$(\Rightarrow)$ Suppose that a 3-partition exists, that is,
$\{a_1,\dots,a_{3r}\} = P_1 \uplus \cdots \uplus P_r$, with $|P_i|=3$ for each $i\in[r]$,  with
$\sum_{j \in P_i} a_j = B$ for each $i \in [r]$. Define $\sigma'$ by
$\sigma_{b_i}' = \{x_j \mid j \in P_i\}$ and
$\sigma_s' = \{h_1,\dots,h_{\lambda B+1}\}$. Then
$v(\sigma_s')=\lambda B+1$ and $v(\sigma_{b_i}')=\lambda B$ for every
$i \in [r]$. Hence the target agents do not envy each other, and each
target agent envies $s$ only up to one helper good. Moreover, exactly
$3r = k$ goods change their owner, so $d(\sigma,\sigma') = k$.

$(\Leftarrow)$ Let $\sigma'$ be an EF1 allocation with
$d(\sigma,\sigma') \le k = 3r$. After at most $3r$ reallocations, agent $s$
still owns at least $\lambda B+1$ goods (including helper goods and partition goods). Therefore, $v\bigl(\sigma_s' \setminus \{\argmax_{g \in \sigma_s'} v(g)\}\bigr)
\geq \lambda B$.
Since $\sigma'$ is EF1, every target agent must receive value at least
$\lambda B$. Thus the total value received by all target agents is at least
$r\lambda B$.

This is possible only if all $3r$ partition goods are moved to the target
agents. To see this, suppose $c\ge 1$ helper goods move to target agents. Then at most $3r-c\le 3r-1$ partition goods move, so by pigeonhole some $b_i$ receives at most $2$ partition goods, giving $v(\sigma'_{b_i})\le 2\lambda(B/2-1)+c = \lambda B - 2\lambda + c$. Since $s$ retains at least one partition good $x_j$, the EF1 threshold for $b_i$ is $v(\sigma_s'\setminus\{x_j\})\ge \lambda B - c$. EF1 then requires $\lambda B-2\lambda+c\ge\lambda B-c$, i.e.\ $c\ge\lambda=3r+1$, contradicting $c\le 3r$. Hence no helper good can be
moved to a target agent. Moreover, the total value of all partition goods is
$r\lambda B$, so each target agent must receive value exactly $\lambda B$. Since each target agent needs value at least $\lambda B$ (from the EF1 condition above) and the total value of all partition goods is exactly $r\lambda B$, the sum constraint forces each $b_i$ to receive exactly $\lambda B$.
As every partition good has value $\lambda a_j$, the goods assigned to each
target agent correspond to a subset of the original numbers summing to $B$. Since $B/4 < a_j < B/2$ for all $j$, any two goods sum to less than $B$ and any four goods sum to more than $B$, so each group must contain exactly $3$ goods.
Hence these subsets form a valid 3-partition.

\end{proof}

\subsection{Binary valuations}

Next, we show that \realloc{} is W[2]-hard for the parameter~$k$ when agents have binary valuations, for $\beta \in \{EF,EFX,EF1\}$.  For this we use a reduction from
\textsc{Hitting Set}.  A \textsc{Hitting Set} instance consists
of a hypergraph~\(\mathcal{H} = (V, \mathcal{E})\), and an integer~\(k\).  The goal
is to decide if there is a \emph{hitting set}~\(S \subseteq V\)
with~\(|{S}| \leq k\) and~\(S \cap e \neq \emptyset\) for all~\(e \in
\mathcal{E}\).  \textsc{Hitting Set} is W[2]-hard for the
parameter~$k$~\cite{DF13}.

\begin{theorem}\label{thm:npbin}
\reallocEFone\ is NP-hard and is $W[2]$-hard parameterized by the
number of reallocations $k$, even when all agents have binary
valuations.
\end{theorem}

\begin{proof}
Given an instance $(\mathcal{H}=(V,\mathcal{E}),k)$ of \textsc{Hitting Set}, we may assume that every hyperedge has size at least two. Indeed, if
$\emptyset \in \mathcal{E}$, then the instance is trivially a no-instance.
Moreover, if $\{v\}\in\mathcal{E}$, then $v$ must be contained in every hitting set; hence we may delete all hyperedges containing $v$, decrease
$k$ by one, and repeat. This preserves equivalence and does not increase the parameter.

We then construct a \reallocEFone{} instance $(A,G,v,\sigma,k)$ as follows.

For each vertex~$v \in V$, there is a good~$g_v \in G$, and for every hyperedge $e\in\mathcal{E}$, we add $|e|-2$ additional goods $g_{e,1},\ldots,g_{e,|e|-2}$.
Next, for each edge~$e \in \mathcal{E}$, we add an agent~$a_e$ to~$A$, as well as an additional agent $a^\star$ and $k$ selector agents $b_1,\ldots,b_k$. 

The valuation functions look as follows:
\begin{align*}
    v_{a_e}(g_{e',i}) =
    \begin{cases}
        1,  \text{ if } e=e',\\
        0,  \text{ otherwise,}
    \end{cases} \qquad 
    v_{a_e}(g_v) =
    \begin{cases}
        1,  \text{ if } v\in e,\\
        0,  \text{ otherwise.}
    \end{cases}
\end{align*}
Lastly, $v_{a^\star}(g)=0$ and $v_{b_j}(g)=0$ for all goods $g\in G$ and all $j\in[k]$.
Additionally, we construct the initial allocation $\sigma$. We set
$\sigma_{a_e}=\{g_{e,i} \mid i\in[|e|-2]\}$
for every $e\in\mathcal{E}$, and
$\sigma_{a^\star}=\{g_v \mid v\in V\}$.
Every selector agent $b_j$ initially receives the empty bundle. Finally,
we set the number of reallocations to $k$. This finishes the construction, which is clearly polynomial.

We claim that there exists a hitting set of size at most $k$ in
$\mathcal{H}$ if and only if the \reallocEFone{} instance has a solution
$\sigma'$ reachable by at most $k$ reallocations.

$(\Rightarrow)$ Let $S$ be a hitting set of size at most $k$ in
$\mathcal{H}$. Consider the allocation $\sigma'$ in which all goods
$g_v$ with $v\in S$ are separately reallocated to distinct selector
agents. This uses at most $k$ reallocations.
Note that initially every hyperedge agent $a_e$ envied $a^\star$ up to
one good, since
$
v_{a_e}(\sigma_{a^\star})=|e|$
 and 
$v_{a_e}(\sigma_{a_e})=|e|-2.
$
After removing the goods corresponding to $S$, we have
$
v_{a_e}(\sigma'_{a^\star})\le |e|-1,
$
as, by definition of hitting set, we removed at least one valuable good from
$a^\star$ for every hyperedge $e$. Furthermore, as we do not remove goods
from the bundles of the hyperedge agents,
$
v_{a_e}(\sigma'_{a_e})=|e|-2.
$
Thus $a_e$ does not envy $a^\star$ up to one good. Since $a^\star$ and all selector agents assign value zero to all goods they do not envy anyone. Moreover, hyperedge agents assign value zero to the allocations of other hyperedge agents, and since each selector agent receives at most one good, no hyperedge agent EF1-envies a selector agent. Therefore $\sigma'$ is an EF1 allocation reachable by at most $k$ reallocations.

$(\Leftarrow)$ Let $\sigma'$ be a solution to the \reallocEFone{} instance.
Let
$S=\{v\in V \mid \sigma'(g_v)\neq a^\star\}$.
Since at most $k$ goods are reallocated, we have $|S|\le k$.
Suppose that $S$ does not form a hitting set. This implies that there exists
a hyperedge $e\in\mathcal{E}$ such that all corresponding vertex goods
remain in $\sigma'_{a^\star}$. Hence
$
v_{a_e}(\sigma'_{a^\star})\ge |e|.
$
For every $e\in\mathcal{E}$, the total value of all goods for
agent $a_e$ is
$
v_{a_e}(G)=2|e|-2.
$
Since $a^\star$ holds goods of value at least $|e|$ for $a_e$, the value
of $a_e$'s bundle is at most
$
v_{a_e}(\sigma'_{a_e})\le |e|-2.
$
As all valuations are binary, removing any single good from
$\sigma'_{a^\star}$ decreases its value for $a_e$ by at most one. Thus,
for every $g\in\sigma'_{a^\star}$,
$
v_{a_e}(\sigma'_{a^\star}\setminus\{g\})
\ge |e|-1
>
|e|-2
\ge v_{a_e}(\sigma'_{a_e}).
$
Thus $a_e$ EF1-envies $a^\star$, contradicting that $\sigma'$ is EF1. Hence $S$ is a hitting set for $\mathcal{H}$.

The reduction preserves $k$ and valuations are binary.
Since \textsc{Hitting Set} is NP-hard and $W[2]$-hard parameterized by $k$,
the result follows.
\qedhere
\end{proof}

We use a similar reduction to prove this for $\beta \in \{EF,EFX\}$.

\begin{theorem}\label{thm:npbinX}
\realloc\ is NP-hard and is $W[2]$-hard parameterized by the number of reallocations $k$, even when all agents have binary valuations for $\beta \in \{EF,EFX\}$.
\end{theorem}

\begin{proof}
Given an instance $(\mathcal{H}=(V,\mathcal{E}),k)$ of \textsc{Hitting Set}, again we may assume that every hyperedge has size at least two.

We construct a \reallocEF{} instance $(A,G,v,\sigma,k)$ as follows.
For each vertex~$v\in V$, there is a good~$g_v\in G$, and for every
hyperedge $e\in\mathcal{E}$, we add $|e|-1$ additional goods
$g_{e,1},\ldots,g_{e,|e|-1}$.
Next, for each hyperedge~$e\in\mathcal{E}$, we add an agent~$a_e$ to~$A$,
as well as an additional agent $a^\star$ and $k$ selector agents
$b_1,\ldots,b_k$.

The valuation functions look as follows:
\begin{align*}
    v_{a_e}(g_{e',i}) =
    \begin{cases}
        1,  \text{ if } e=e',\\
        0,  \text{ otherwise,}
    \end{cases}
    \qquad
    v_{a_e}(g_v) =
    \begin{cases}
        1,  \text{ if } v\in e,\\
        0,  \text{ otherwise.}
    \end{cases}
\end{align*}
Lastly, $v_{a^\star}(g)=0$ and $v_{b_j}(g)=0$ for all goods $g\in G$
and all $j\in[k]$.

Additionally, we construct the initial allocation $\sigma$. We set $\sigma_{a_e}=\{g_{e,i}: i\in[|e|-1]\}$ 
for every $e\in\mathcal{E}$, and $\sigma_{a^\star}=\{g_v \mid v\in V\}$.
Every selector agent $b_j$ initially receives the empty bundle. Finally, we set the number of reallocations to $k$. This finishes the construction, which is clearly polynomial.

We claim that there exists a hitting set of size at most $k$ in $\mathcal{H}$ if and only if the \reallocEF{} instance has a solution
$\sigma'$ reachable by at most $k$ reallocations.

$(\Rightarrow)$ Let $S$ be a hitting set of size at most $k$ in
$\mathcal{H}$. Consider the allocation $\sigma'$ in which all goods
$g_v$ with $v\in S$ are separately reallocated to distinct selector
agents. This uses at most $k$ reallocations.
Note that initially every hyperedge agent $a_e$ envied $a^\star$, since $v_{a_e}(\sigma_{a^\star})=|e|$ and $v_{a_e}(\sigma_{a_e})=|e|-1$.

After removing the goods corresponding to $S$, we have $v_{a_e}(\sigma'_{a^\star})\le |e|-1$,
as, by definition of hitting set, we removed at least one valuable good from $a^\star$ for every hyperedge $e$. Furthermore, as we do not remove
goods from the bundles of the hyperedge agents,  $v_{a_e}(\sigma'_{a_e})=|e|-1$.
Thus $a_e$ does not envy $a^\star$.

Since $a^\star$ and all selector agents assign value zero to all goods,
they do not envy anyone. Moreover, hyperedge agents assign value zero to
the allocations of other hyperedge agents, and since each selector agent
receives at most one good, every hyperedge agent values the bundle of any
selector agent by at most one. As every hyperedge has size at least $2$, we have $v_{a_e}(\sigma'_{a_e})=|e|-1\ge 1$ so no envy is created toward a selector agent. Therefore $\sigma'$ is an EF allocation reachable by at most $k$ reallocations.

$(\Leftarrow)$ Let $\sigma'$ be a solution to the \reallocEF{} instance.
Let $S=\{v\in V \mid \sigma'(g_v)\neq a^\star\}$.
Since at most $k$ goods are reallocated, we have $|S|\le k$.
Suppose that $S$ does not form a hitting set. This implies that there exists a hyperedge $e\in\mathcal{E}$ such that all corresponding vertex
goods remain in $\sigma'_{a^\star}$. Hence
$v_{a_e}(\sigma'_{a^\star})\ge |e|$.
For every $e\in\mathcal{E}$, the total value of all goods for agent
$a_e$ is $v_{a_e}(G)=2|e|-1$.
Since $a^\star$ holds goods of value at least $|e|$ for $a_e$, the value of $a_e$'s bundle is at most
$v_{a_e}(\sigma'_{a_e})\le |e|-1$.
Thus $a_e$ envies $a^\star$, contradicting that $\sigma'$ is EF. Hence $S$ is a hitting set for $\mathcal{H}$.

The reduction preserves the parameter $k$ and valuations are binary.
Since \textsc{Hitting Set} is NP-hard and $W[2]$-hard parameterized by
$k$, the result follows for \reallocEF{}.
By \cref{lem:ef-efx}, the same result follows for \reallocEFX{}.
\end{proof}

The preceding hardness results show that identical valuations and binary valuations alone are not sufficient for tractability in general. However, when both restrictions are imposed simultaneously, the problem becomes tractable.

\subsection{Identical and Binary}
Combining both restrictions, the optimal bundle structure is determined by value-$1$ counts alone, yielding polynomial time algorithms.

\begin{theorem}\label{thm:polyidx-ef1}
    When agents have identical binary valuations, \realloc{} is solvable in $\Oh(n+m)$ time for $\beta \in \{EF, EF1\}$.
\end{theorem}

\begin{proof}
    Let $\sigma=(\sigma_1,\dots,\sigma_n)$ be the initial allocation and let $k$
    be the allowed number of reallocations. Since valuations are identical and
    binary, each good has value either $0$ or $1$ for every agent. Let $m_1$ be
    the number of value-$1$ goods, and define
     $q=\left\lfloor \frac{m_1}{n}\right\rfloor$,
        $r=m_1-nq$.

    We start by proving the claim for EF1. Under identical binary valuations, an allocation is EF1
    if and only if all bundle values differ by at most one. Indeed, if
    $v(X_j)\ge v(X_i)+2$, then removing any single good from $X_j$ decreases its
    value by at most one, so agent $i$ still envies agent $j$. Conversely, if all
    bundle values differ by at most one, then whenever $i$ envies $j$, we have
    $v(X_j)=v(X_i)+1$. Thus $X_j$ contains a value-$1$ good, and removing such a
    good eliminates the envy.

    Hence every EF1 allocation must have exactly $r$ agents with value $q+1$,
    and all remaining agents with value $q$.

    For each agent $a_i$, let $s_i=v(\sigma_i)$. If $H\subseteq A$ is the set of
    agents chosen to have final value $q+1$, with $|H|=r$, define the target
    value
    \[
        t_i =
        \begin{cases}
            q+1, & a_i\in H,\\
            q, & a_i\notin H .
        \end{cases}
    \]
    For this fixed choice of $H$, the minimum number of reallocations of
    value-$1$ goods needed is $C(H)=\sum_{i=1}^n \max\{0,t_i-s_i\}$.
    This many moves are necessary, since each deficient agent $a_i$ must receive $t_i-s_i$ value-$1$ goods, and one reallocation can increase the value of at most one deficient bundle by one. They are also sufficient, because the total
    target value equals the total initial value, so the total surplus equals the total deficit.

    It remains to choose $H$ minimizing $C(H)$. Let $B=\sum_{i=1}^n \max\{0,q-s_i\}$. This is the cost of bringing every agent to value at least $q$. Choosing
    agent $a_i$ to be in $H$ instead of outside $H$ increases the cost by
    $$\Delta_i=\max\{0,q+1-s_i\}-\max\{0,q-s_i\}.$$
    Thus $\Delta_i=1$ if $s_i\le q$, and $\Delta_i=0$ otherwise. Therefore an
    optimal set $H$ is obtained by choosing $r$ agents with minimum $\Delta_i$, which can be done by sorting.

    If the resulting minimum value $B+\sum_{a_i\in H}\Delta_i$ exceeds $k$, we return \textsc{No}. Otherwise, we move value-$1$ goods from agents above their target value to agents below their target value until every agent reaches its target. The number of moves is exactly $B+\sum_{a_i\in H}\Delta_i\le k$.

The resulting allocation has exactly $r$ agents of value $q+1$ and all other
agents of value $q$, and hence is EF1 by the characterization above. Computing
all bundle values takes $\Oh(m)$ time. Computing $q,r,B$ and all values
$\Delta_i$ takes $\Oh(n)$ time. Since $\Delta_i\in\{0,1\}$ for every agent, an
optimal set $H$ can be found in linear time by first choosing agents with
$\Delta_i=0$ and then, if necessary, arbitrary agents with $\Delta_i=1$ until
$|H|=r$. Constructing the target allocation performs at most $m$ reallocations. Hence the algorithm runs in $\Oh(m+n)$ time.

For EF, let $m_1$ be the number of value-$1$ goods. If $m_1$ is not divisible
by $n$, then no EF allocation exists, since identical valuations require all
agents to receive the same value. Hence we may assume that $q=\frac{m_1}{n}$ is an integer.

For each agent $a_i$, let $s_i=v(\sigma_i)$. In any EF allocation, every agent must have value exactly $q$. Thus every agent with $s_i<q$ must receive
$q-s_i$ value-$1$ goods. Since one reallocation can increase the value of at
most one deficient bundle by one, at least $\sum_{i \in [n]:s_i<q} (q-s_i)$
reallocations are necessary. Conversely, this many reallocations are sufficient:
move value-$1$ goods from agents with value greater than $q$ to agents with
value less than $q$ until every agent has value exactly $q$. The total surplus
equals the total deficit because the total value is $nq$.

Therefore, the minimum number of reallocations needed to reach an EF allocation is $\sum_{i \in [n]:s_i<q} (q-s_i)$.
If this number is greater than $k$, return \textsc{No}; otherwise return
\textsc{Yes} and perform the corresponding reallocations.

Computing $m_1$ and all bundle values takes $\Oh(m)$ time, and computing the above sum takes $\Oh(n)$ time. The construction performs at most $m$ reallocations. Hence the algorithm runs in $\Oh(m+n)$ time. 
\end{proof}

\begin{theorem}\label{thm:polyidx}
   When agents have identical binary valuations, \reallocEFX{} is solvable in $\Oh(m+n\log n)$ time.
\end{theorem}

\begin{proof}
    We proceed similarly to \cref{thm:polyidx-ef1}, but additionally account for the fact that, in an EFX allocation, every bundle of value $q+1$ must contain only value-$1$ goods. 
    
    Let $\sigma=(\sigma_1,\dots,\sigma_n)$ be the given allocation, and let $k$ be
    the allowed number of reallocations. Since valuations are identical and
    binary, every good has value either $0$ or $1$ for every agent. Let $m_1$
    denote the number of value-$1$ goods, and let
    \[
        q\coloneqq\left\lfloor \frac{m_1}{n}\right\rfloor,
        \qquad
        r\coloneqq m_1-nq .
    \]
    Thus any allocation whose bundle values lie in $\{q,q+1\}$ has exactly
    $r$ agents of value $q+1$.

    We first characterize the target structure of an EFX allocation. In any EFX
    allocation under identical binary valuations, every agent must receive value
    either $q$ or $q+1$. Indeed, if some agent received value at most $q-1$,
    then, since the total value is $m_1$, some other agent would receive value
    at least $q+1$. Removing any value-$1$ good from the latter agent's bundle
    would still leave value at least $q$, so the former agent would EFX-envy
    them.

    Moreover, every bundle of value $q+1$ must contain only value-$1$ goods.
    Otherwise, if such a bundle contained a value-$0$ good, then an agent with
    value $q$ would still envy it after that value-$0$ good was removed.

    Hence an EFX allocation has the following structure: exactly $r$ agents have
    value $q+1$, all other agents have value $q$, and every value-$(q+1)$ bundle
    contains no value-$0$ good.

    \smallskip\noindent\emph{Algorithm.}
    For every agent $a_i$, let $x_i$
be the number of value-$1$ goods initially in $\sigma_i$, and let $z_i$ be the
number of value-$0$ goods initially in $\sigma_i$.

Now fix a set $H\subseteq A$ of size $r$. We interpret $H$ as the set of agents
that will have final value $q+1$. Thus the target value of agent $a_i$ is
\[
    t_i(H):=
    \begin{cases}
        q+1, & a_i\in H,\\
        q, & a_i\notin H.
    \end{cases}
\]

For this fixed choice of $H$, the minimum number of value-$1$ reallocations
needed to reach the target values is
$\sum_{i=1}^n \max\{0,t_i(H)-x_i\}$.
This many reallocations are necessary because each such reallocation can
increase the value of at most one deficient bundle by one. They are also
sufficient because the total target value equals the total initial value,
namely $m_1$, so the total surplus of value-$1$ goods equals the total deficit.

In addition, if an agent $a_i\in H$ initially holds $z_i$ value-$0$ goods, then
all these goods must be moved away. Indeed, agents in $H$ are precisely the
agents that must end with value $q+1$, and we have shown above that a
value-$(q+1)$ bundle cannot contain any value-$0$ good in an EFX allocation.
Thus each of these $z_i$ goods contributes one necessary reallocation.

For the fixed choice of $H$, the value-$1$ reallocations and the value-$0$
reallocations account for disjoint goods. Therefore the exact minimum number of reallocations for this fixed
choice of $H$ is
\[
    C(H)
    :=
    \sum_{i\notin H}\max\{0,q-x_i\}
    +
    \sum_{i\in H}\max\{0,q+1-x_i\}
    +
    \sum_{i\in H}z_i .
\]
We now rewrite this cost in a form that separates the part independent of $H$
from the part depending on $H$. Since every agent must have final value at least $q$, every feasible EFX target requires the baseline cost
$B:=\sum_{i=1}^n \max\{0,q-x_i\}$. 
Then
\[
\begin{aligned}
    C(H)
    &=
    \sum_{i=1}^n \max\{0,q-x_i\}
    +
    \sum_{i\in H}
    \Bigl(
        \max\{0,q+1-x_i\}
        -
        \max\{0,q-x_i\}
        +
        z_i
    \Bigr)  \\
    &=
    B+\sum_{i\in H}\gamma_i,
\end{aligned}
\]
where
\[
    \gamma_i
    :=
    \max\{0,q+1-x_i\}
    -
    \max\{0,q-x_i\}
    +
    z_i .
\]
Since $x_i$ is integral, the difference
$\max\{0,q+1-x_i\}-\max\{0,q-x_i\}$
is equal to $1$ if $x_i\le q$, and equal to $0$ otherwise. Hence
\[
    \gamma_i=\mathbf 1[x_i\le q]+z_i .
\]
Thus $\gamma_i$ is the marginal cost of choosing $a_i$ as one of the final
value-$(q+1)$ agents. This
marginal cost consists of two parts: if $x_i\le q$, then $a_i$ needs one
additional value-$1$ good to reach $q+1$ rather than merely $q$; moreover, all
$z_i$ initially held value-$0$ goods must be removed from $a_i$'s bundle.

It remains to minimize $C(H)$ over all sets $H$ of size $r$. Since $B$ is
independent of $H$, this is equivalent to minimizing $\sum_{i\in H}\gamma_i$
over all sets $H$ of size $r$. Therefore an optimal choice is obtained by
letting $H$ be the set of $r$ agents with smallest values of $\gamma_i$.

To see this formally, fix any set $H$ of size $r$. The value $C(H)$ is the
minimum cost of reaching an EFX allocation whose value-$(q+1)$ agents are
exactly those in $H$. Thus the exchange argument compares feasible EFX target
structures.
Suppose that $H$ does not consist of the $r$ smallest values of $\gamma_i$.
Then there exist agents $a_i\in H$ and $a_j\notin H$ such that $\gamma_j\le \gamma_i$.
Let $H'=(H\setminus\{a_i\})\cup\{a_j\}$.
Then $H'$ is again a feasible choice of the value-$(q+1)$ agents, and $C(H')=C(H)-\gamma_i+\gamma_j\le C(H)$.
Repeating this exchange eventually yields a set consisting of the $r$ smallest values of $\gamma_i$, without increasing the cost. Hence choosing the $r$ smallest $\gamma_i$ values minimizes $C(H)$ over all EFX target structures.

Let $H$ be such an optimal set. If $B+\sum_{a_i\in H}\gamma_i>k$,
then every possible choice of the $r$ final value-$(q+1)$ agents requires more
than $k$ reallocations. Hence no EFX allocation is reachable within the budget.
Otherwise, we construct one as follows.

For every agent $a_i$, set the target value to $q+1$ if $a_i\in H$, and to
$q$ otherwise. Move value-$1$ goods from agents whose current value exceeds
their target value to agents whose current value is below their target value
until all agents reach their target values. This uses exactly $\sum_{i=1}^n \max\{0,t_i(H)-x_i\}$
reallocations. Then move every value-$0$ good held by an agent in $H$ to some
agent outside $H$. This is possible since $r<n$ by definition of
$r=m_1-nq$. These moves use exactly $\sum_{i\in H} z_i$ reallocations. Hence the total number of reallocations is exactly $C(H)$.

The resulting allocation is EFX. Agents in $H$ have value $q+1$ and hold only
value-$1$ goods, while all other agents have value $q$. Thus no agent envies an
agent of weakly smaller value. If an agent of value $q$ compares herself with an agent of value $q+1$, then removing any good from the latter bundle leaves value $q$, because that bundle contains only value-$1$ goods. Hence there is no EFX envy.

It remains to analyze the running time. Computing all values $x_i$ and all
numbers $z_i$ takes $\mathcal O(m)$ time. Computing the values $\gamma_i$ takes $\mathcal O(n)$ time. Selecting the $r$ smallest $\gamma_i$ values can be done by sorting in $\mathcal O(n\log n)$ time. The construction moves at most $m$ goods. Hence the total running time is $\mathcal O(m+n\log n)$.
\end{proof}

\section{Parameterized Algorithms}
We study \realloc{} parameterized by $n$, $m$, $k$, and $v_{\max}$; \Cref{fig:params-overview} summarizes the landscape.

\begin{lemma}\label[lemma]{lem:red-ef1-efx}
Each instance $(A,G,v,\sigma,k)$ of \realloc{}
can be reduced in polynomial time to an equivalent instance
$(A',G,v|_{A'},\sigma|_{A'},k)$ of \realloc{}
with $|A'| \leq \mathcal{O}(m^3)$, for
$\beta\in\{\mathrm{EF1},\mathrm{EFX}\}$.
\end{lemma}

\begin{proof} Let $n\coloneqq |A|$.  If $n \leq m + m(m+1)+m^2(m+1)$, there is nothing to do. Hence assume otherwise. The set of agents $A'$ is selected as follows. We first mark every agent who is assigned at least one good in the initial allocation $\sigma$; there are at most
$m$ such agents. Then, for each good $x_j$, we mark the $m+1$ agents with the
highest valuation for $x_j$, breaking ties arbitrarily. Moreover, for each pair
$x_i,x_j$ of distinct goods, we mark $m+1$ agents who evaluate both $x_i$ and
$x_j$ positively, or all such agents if there are at most $m$ of them. After
doing this for all goods and pairs of goods, we delete all unmarked agents.

Clearly, the construction can be carried out in polynomial time. Moreover, we mark at most $m$ agents due to the initial allocation, at most $m(m+1)$ agents for single goods, and at most $m^2(m+1)$ agents for pairs of goods. Hence  $|A'| \leq m + m(m+1)+m^2(m+1)=\mathcal{O}(m^3)$.

We now prove equivalence. 

First suppose that the original instance admits an
allocation $\pi$ reachable from $\sigma$ by at most $k$ reallocations that is EF1, respectively EFX. If every good in $\pi$ is assigned to a marked agent, then $\pi$ is also a feasible solution for the reduced instance. 
Otherwise, suppose some good $x_j$ is assigned in $\pi$ to an unmarked agent.
By construction, there are $m+1$ marked agents whose value for $x_j$ is at least
that of the unmarked agent. Since there are only $m$ goods in total, at least
one of these marked agents receives no good in $\pi$. Reassign $x_j$ to such an
empty marked agent. This does not increase the number of reallocated goods.
Moreover, the new holder receives only the single good $x_j$, so no agent can
EF1- or EFX-envy them, since after removing $x_j$, her bundle is empty. Repeating this replacement for all goods assigned to unmarked agents yields a feasible solution for the reduced instance.

Conversely, suppose that the reduced instance admits a solution $\pi$ reachable
from $\sigma_{|A'}$ by at most $k$ reallocations. We show that the same allocation is EF1, respectively EFX, in the original instance. Assume towards a
contradiction that this is not the case. Then there is an unmarked agent $a_u$
who EF1-envies, respectively EFX-envies, the bundle assigned to some marked
agent $a_v$.
Since $a_u$ receives no good in $\pi$, we have $v_u(\pi_u)=0$. We distinguish
two cases.

First consider EF1. Since $a_u$ EF1-envies $a_v$, for every good
$g\in\pi_v$ we have $v_u(\pi_v\setminus\{g\})>0$. Hence $\pi_v$ contains
two goods $x_i,x_j$ with $v_u(x_i)>0$ and $v_u(x_j)>0$. By construction,
for the pair $x_i,x_j$, we marked $m+1$ agents who value both goods
positively, unless there were at most $m$ such agents in total. Since
$a_u$ is unmarked and values both goods positively, the latter case is
impossible. Thus one of these $m+1$ marked agents receives no good in
$\pi$. This empty marked agent also EF1-envies $a_v$, because after
removing any one good from $\pi_v$, at least one of $x_i,x_j$ remains.

Now consider EFX. Since $a_u$ EFX-envies $a_v$, there exists a good
$y\in\pi_v$ such that $v_u(\pi_v\setminus\{y\})>0$. Hence some good
$x_i\in\pi_v\setminus\{y\}$ satisfies $v_u(x_i)>0$. By construction, for
$x_i$ we marked $m+1$ agents with highest valuation for $x_i$; since
$a_u$ is unmarked and values $x_i$ positively, all these marked agents
also value $x_i$ positively. One of them is empty in $\pi$, and after
removing the same good $y$, it still positively values the remaining
bundle. Hence it EFX-envies $a_v$.

In either case, a marked empty agent violates EF1, respectively EFX,
contradicting feasibility of $\pi$ in the reduced instance.
Therefore the reduced instance is equivalent to the original instance.
\qedhere
\end{proof}

Combining the above reduction with exhaustive search yields the following fixed-parameter tractability result.

\begin{theorem}\label{cor:fptm}
\realloc{} is fixed-parameter tractable with respect to $m$ and XP with respect to $k$ for every $\beta\in\{\mathrm{EF},\mathrm{EF1},\mathrm{EFX}\}$.
\end{theorem}

\begin{proof}
For $\beta\in\{\mathrm{EF1},\mathrm{EFX}\}$, apply
\cref{lem:red-ef1-efx}. The resulting equivalent instance has
$\mathcal{O}(m^3)$ agents. 
We can now enumerate all possible reallocations of at most $k$ goods. For
each reallocated good, there are $m$ choices for the good and $n$
choices for the receiving agent.  Hence, the total number of
possibilities is at most $(m n)^k$. For each such guess, we can verify feasibility in polynomial time. Thus \realloc{} can be solved in time $\mathcal{O}((nm)^k)$.
Since $k\leq m$ and, after the reduction, $n=\mathcal{O}(m^3)$, this running time is bounded by a function of $m$ times a polynomial in the input size.

For EF, observe that if more than one agent who values some good positively
receives no good, then no EF allocation exists. Thus, after discarding agents
who value all goods at zero, either the number of relevant agents is already at
most $m$, or the instance is a \textsc{No}-instance. Hence \reallocEF{} is also fixed-parameter tractable with respect to $m$.
\qedhere
\end{proof}

\begin{theorem}\label{thm:fpt-nvmax}
For $\beta \in \{EF, EFX\}$, \realloc{} is solvable in $\bigl(2n(v_{\max}+1)^n\bigr)^{O\!\left(n(v_{\max}+1)^n\right)}
\cdot \mathrm{poly}(m,\,n,\,\log v_{\max})$ time.
\end{theorem}

\begin{proof}
We say that two goods $g,g'$ have the same type if $v_i(g) = v_i(g')$ for all $a_i \in A$. Since goods of the same type are interchangeable for every valuation-based fairness criterion, we may represent any allocation $\sigma$ by its \emph{type-count vector}
$\mathbf{c} = (c_{t,i})_{t\in[\tau],a_i\in A}$, where $\tau$ is the number of good types and $c_{t,i} \in \mathbb{Z}_{\geq 0}$ denotes the number of type-$t$ goods held by agent $a_i$.

Let $T := (v_{\max}+1)^n$ be an upper bound on the number of good types. 
For a type $t$ and an agent $a_i$, let $c_{t,i}$ denote the number of goods
of type $t$ initially allocated to $a_i$, and let $C_t := \sum_{a_i\in A} c_{t,i}$ be the total number of goods of type $t$.

We first present an ILP formulation for the EFX case.
For every $t\in[T]$ and $a_i\in A$, the formulation uses the following variables:
$c'_{t,i}\in \mathbb{Z}_{\geq 0}$,
$b_{t,i}\in\{0,1\}$ and 
$s_{t,i}\in\mathbb{R}_{\geq 0}$.
Here $c'_{t,i}$ is the number of type-$t$ goods assigned to agent $a_i$ in the
new allocation, $b_{t,i}$ indicates whether agent $a_i$ receives at least one
type-$t$ good, and $s_{t,i}$ is an auxiliary surplus variable used to linearise the term $\max\{c_{t,i}-c'_{t,i},0\}$, the number of type-$t$ goods removed from agent $a_i$.
Let $M := m \cdot v_{\max}$. We use the following constraints:
\begin{alignat}{3}
    &\text{Find} \quad && c'_{t,i} \in \mathbb{Z}_{\geq 0},\;
      b_{t,i} \in \{0,1\},\; s_{t,i} \in \mathbb{R}_{\geq 0}
    \quad &\forall\, t \in [T],\, a_i \in A, \notag\\
    &\text{s.t.} && \sum_{a_i \in A} c'_{t,i} = C_t \quad &\forall\, t\in [T], \tag{conservation}\\
    &&& c'_{t,i} \leq C_t\, b_{t,i}, \quad c'_{t,i} \geq b_{t,i}
      \quad &\forall\, t\in [T],a_i \in A, \tag{presence}\\
    &&& s_{t,i} \geq c_{t,i} - c'_{t,i} \quad &\forall\, t\in [T],a_i \in A,
      \tag{surplus}\\
    &&& \sum_{t,a} s_{t,i} \leq k,
      \tag{budget}\\
    &&& v_i(\sigma'_i) \geq v_i(\sigma'_j) - v_i(g_t) - M(1-b_{t,j})
        \quad &\forall\,t\in [T], a_i,a_j \in A, \tag{EFX}
\end{alignat}
where $C_t = \sum_{a_i \in A} c_{t,i}$,
$v_i(\sigma'_i) = \sum_{t\in [T]} c'_{t,a}\, v_i(g_t)$, and $g_t$ denotes any fixed representative good of type $t$.

The \emph{conservation} constraints preserve the total number of goods of each type. The \emph{presence} constraints encode the indicators $b_{t,i}$: together, they imply
$b_{t,i}=1 \Leftrightarrow c'_{t,i}\geq 1$. Indeed, if $b_{t,i}=0$, then $c'_{t,i}\leq 0$, and since
$c'_{t,i}\geq 0$, we get $c'_{t,i}=0$. Conversely, if $c'_{t,i}=0$, then
$c'_{t,i}\geq b_{t,i}$ implies $b_{t,i}=0$.
The \emph{surplus} variable
$s_{t,i}$ linearises the loss in type-$t$ goods for agent $a_i$: because
$s_{t,i} \geq c_{t,i}-c'_{t,i}$ and $s_{t,i} \geq 0$, it satisfies
$s_{t,i} \geq \max(c_{t,i}-c'_{t,i},0)$, so the \emph{budget} constraint
$\sum_{t,i} s_{t,i} \leq k$ is equivalent to bounding the total number of goods removed from their initial agents by $k$.  The \emph{EFX} constraints enforce envy-freeness up to the removal of any good in the envied agent's bundle. If $b_{t,j}=1$, then agent $a_j$ receives at least one good of type $t$, and the constraint becomes $v_i(\sigma'_i)
    \geq
    v_i(\sigma'_j)-v_i(g_t)$.
This is exactly the EFX condition for removing a type-$t$ good from $a_j$'s bundle, as all goods of the same type have the same value for every
agent. If $b_{t,j}=0$, then agent $a_j$ receives no good of type $t$, and
the big-$M$ term relaxes the constraint. Since all good values are
nonnegative and at most $v_{\max}$, we have $v_i(\sigma'_j)\leq m \cdot v_{\max}=M$,
so the relaxed constraint is automatically satisfied. 

\noindent\textbf{Correctness}
We show the ILP is feasible if and only if \reallocEFX{} has a solution.

\smallskip
\noindent$(\Rightarrow)$~
Given a feasible ILP solution $(\mathbf{c}',\mathbf{b},\mathbf{s})$, construct
$\sigma'$ by assigning exactly $c'_{t,i}$ goods of type $t$ to agent $a_i$.
The conservation constraints ensure that this yields a valid partition of all
goods. Because $s_{t,i} \geq c_{t,i}-c'_{t,i}$ and $s_{t,i}\geq 0$, we have
$\max\{c_{t,i}-c'_{t,i},0\} \leq s_{t,i}$
for every $t,i$. Hence $d(\sigma,\sigma')
    =
    \sum_{t,i}\max\{c_{t,i}-c'_{t,i},0\}
    \leq
    \sum_{t,i} s_{t,i}
    \leq k$.
It remains to show that $\sigma'$ is EFX. For any pair $a_i,a_j$ and any
$g\in\sigma'_j$ of type $t$, we have $c'_{t,j}\geq 1$, so $b_{t,j}=1$ by
the presence constraints. The EFX constraint therefore gives $v_i(\sigma'_i)
    \geq
    v_i(\sigma'_j)-v_i(g_t)
    =
    v_i(\sigma'_j)-v_i(g)$,
where the equality follows because $g$ and $g_t$ have the same type. Hence
$a_i$ does not envy $a_j$ after removing $g$, and so $\sigma'$ is EFX.

\smallskip
\noindent$(\Leftarrow)$~
Let $\sigma' = (\sigma'_1,\dots,\sigma'_n)$ be an EFX allocation with
$d(\sigma,\sigma') \leq k$. Set
\[
    c'_{t,i}
    :=
    |\{g \in \sigma'_i \mid \mathrm{type}(g)=t\}|,
    \quad
    b_{t,i}
    :=
    \mathbf{1}[c'_{t,i}\geq 1], \quad
    s_{t,i}
    :=
    \max\{c_{t,i}-c'_{t,i},0\}.
\]
Because $\sigma'$ partitions all goods, $\sum_i c'_{t,i}=C_t$ for every
type $t$, so conservation holds. The definition of $b_{t,i}$ directly
satisfies both presence constraints. By construction,
$s_{t,i}\geq c_{t,i}-c'_{t,i}$ and $s_{t,i}\geq 0$, so the surplus
constraints hold. Moreover,
\[
    \sum_{t,i} s_{t,i}
    =
    \sum_{t,i}\max\{c_{t,i}-c'_{t,i},0\}
    =
    d(\sigma,\sigma')
    \leq k,
\]
so the budget constraint holds.

Finally, fix any $a_i,a_j$ and any type $t$. If $b_{t,j}=0$, the EFX constraint
is relaxed by the big-$M$ term and is automatically satisfied. If $b_{t,j}=1$, then $\sigma'_j$ contains some good $g$ of type $t$. Since $\sigma'$ is EFX,  $v_i(\sigma'_i)
    \geq
    v_i(\sigma'_j)-v_i(g)
    =
    v_i(\sigma'_j)-v_i(g_t)$,
so the EFX constraint holds. Hence the ILP is feasible.

\smallskip
\noindent\textbf{Running time.}
The ILP has $p = 2Tn = 2n(v_{\max}+1)^n$ integer variables, namely the variables $c'_{t,i}$ and $b_{t,i}$. 
The number of constraints and the encoding size are polynomial in
$m,n,\log v_{\max}$, up to the explicit dependence on $T$. By Lenstra~\cite{Lenstra83}, feasibility of an ILP with $p$ integer variables and input encoding size $L$ can be decided in time $p^{O(p)}\cdot\mathrm{poly}(L)$.
This gives an overall running time of $$\bigl(2n(v_{\max}+1)^n\bigr)^{O(n(v_{\max}+1)^n)}\cdot \mathrm{poly}(m,n,\log v_{\max}).$$
Thus \reallocEFX{} is fixed-parameter tractable with respect to $v_{\max}$ and $n$.
It remains to consider the EF case. By \cref{lem:ef-efx}, the above EFX
formulation already implies fixed-parameter tractability for \reallocEF{}.

However, in the EF case we can use a smaller ILP, since no good-removal condition has to be encoded. In particular, the indicator variables $b_{t,i}$ and the presence constraints are unnecessary.

We keep only the variables $c'_{t,i}\in \mathbb{Z}_{\geq 0}$
    and 
    $s_{t,i}\in\mathbb{R}_{\geq 0}$
for every $t\in[T]$ and $a_i\in A$, together with the \emph{conservation}, \emph{surplus},
and \emph{budget} constraints. We replace the EFX constraints by the following EF constraints:
\begin{alignat}{3}
    &\sum_{t\in[T]} c'_{t,i}v_i(g_t)
        \geq
        \sum_{t\in[T]} c'_{t,j}v_i(g_t)
        \quad &&\forall\, a_i,a_j\in A.
        \tag{EF}
\end{alignat}
Here, as before, $g_t$ denotes an arbitrary representative good of type $t$.

The conservation constraints ensure that the resulting type-count vector
corresponds to a valid allocation of all goods. The surplus and budget
constraints enforce $d(\sigma,\sigma')
    =
    \sum_{t,i}\max\{c_{t,i}-c'_{t,i},0\}
    \leq k$,
exactly as in the EFX case. Finally, the constraints (EF) state
precisely that no agent envies any other agent under the allocation induced by
$\mathbf{c}'$. Hence every feasible ILP solution induces an EF allocation
within distance at most $k$ from $\sigma$, and conversely every such allocation
yields a feasible ILP solution by setting
$s_{t,i}=\max\{c_{t,i}-c'_{t,i},0\}$.

This ILP has only $p=Tn=n(v_{\max}+1)^n$ integer variables, namely the variables $c'_{t,i}$. Its encoding size is
polynomial in $m,n,\log v_{\max}$, up to the explicit dependence on $T$.
By Lenstra~\cite{Lenstra83}, feasibility can therefore be decided
in time
\[
    \bigl(n(v_{\max}+1)^n\bigr)^{O(n(v_{\max}+1)^n)}
    \cdot
    \mathrm{poly}(m,n,\log v_{\max}).
\] \end{proof}

We use a similar ILP to prove the following. 
\begin{theorem}\label{thm:fpt-ef1-nvmax}
\reallocEFone{} is solvable in $\bigl(n(n+2)(v_{\max}+1)^n\bigr)^{O(n^2(v_{\max}+1)^n)}\cdot
\mathrm{poly}(m,\,n,\,\log v_{\max})$ time.
\end{theorem}

\begin{proof}
We use the same good-type reduction from Theorem~\ref{thm:fpt-nvmax}. 
We just replace the EFX constraints by
constraints encoding EF1 via explicit witness variables.

For each ordered pair $i\neq j$, EF1 requires either that $a_i$ does not envy $a_j$, or that there exists a good $g\in\sigma'_j$ whose removal eliminates the envy: $v_i(\sigma'_i) \geq v_i(\sigma'_j)-v_i(g)$.
Since all goods of the same type have identical values for every agent, it suffices to record the type of such a witness good.
We introduce binary variables $w_{t,i,j}\in\{0,1\}$ for all $t\in\{0,\dots,T\}$,\ $i\neq j$.
For $t\in[T]$, the variable $w_{t,i,j}=1$ means that a type-$t$ good in $\sigma'_j$ is chosen as the EF1 witness for the pair $(i,j)$. The dummy
type $t=0$ represents the case where no good can be removed; we set $v_i(g_0)\coloneqq 0$ for every agent $a_i$.
We add the following constraints:
\begin{alignat}{3}
    && \sum_{t=0}^{T} w_{t,i,j} &= 1
        \quad && \forall\, i\neq j, \tag{witness-select}\\
    && w_{t,i,j} &\leq b_{t,j}
        \quad && \forall\, t\in[T],\ i\neq j, \tag{witness-valid}\\
    && v_i(\sigma'_i) &\geq v_i(\sigma'_j)-v_i(g_t)-M(1-w_{t,i,j})
        \quad && \forall\, t\in\{0,\dots,T\},\ i\neq j. \tag{EF1}
\end{alignat}
The \emph{witness-select} constraints choose exactly one witness type for each ordered pair $(i,j)$. The \emph{witness-valid} constraints ensure that a real witness type $t\in[T]$ can be selected only if agent $a_j$ actually receives a good of type $t$. The dummy type $0$ handles the case in which $a_j$'s bundle is empty, and hence there is no good whose removal can serve
as a witness. The \emph{EF1} constraints enforce the corresponding inequality. If $w_{t,i,j}=1$, the constraint becomes $v_i(\sigma'_i) \geq v_i(\sigma'_j)-v_i(g_t)$.
For $t\in[T]$, this says that removing a type-$t$ good from $a_j$'s bundle eliminates $a_i$'s envy. For $t=0$, since $v_i(g_0)=0$, it says $v_i(\sigma'_i) \geq v_i(\sigma'_j)$,
so $a_i$ does not envy $a_j$ even without removing a good. If
$w_{t,i,j}=0$, the big-$M$ term relaxes the constraint. Since
$v_i(\sigma'_j)\leq m\cdot v_{\max}=M$ and valuations are nonnegative, the relaxed constraint is automatically satisfied.

\medskip
\noindent\textbf{Correctness.}
We show the extended ILP is feasible if and only if \reallocEFone{} has a solution.

\smallskip
\noindent$(\Rightarrow)$~
Given a feasible ILP solution, construct $\sigma'$ from $\mathbf{c}'$ as in
Theorem~\ref{thm:fpt-nvmax}; the bound $d(\sigma,\sigma')\leq k$ follows by
the same surplus-budget argument. For each pair $i\neq j$,
\emph{witness-select} guarantees a unique witness type
$t^*\in\{0,\dots,T\}$ with $w_{t^*,i,j}=1$. If $t^*=0$, then the \emph{EF1} constraint gives $v_i(\sigma'_i)\geq v_i(\sigma'_j)$, so $a_i$ does not envy $a_j$. If $t^*\in[T]$, then \emph{witness-valid}
ensures $b_{t^*,j}=1$, meaning $\sigma'_j$ contains at least one good $g$
of type $t^*$. The \emph{EF1} constraint gives $v_i(\sigma'_i)
    \geq
    v_i(\sigma'_j)-v_i(g_{t^*})
    =
    v_i(\sigma'_j)-v_i(g)$, confirming that $a_i$ does not EF1-envy $a_j$. Since this holds for every
ordered pair, $\sigma'$ is EF1.

\smallskip
\noindent$(\Leftarrow)$~
Let $\sigma'$ be an EF1 allocation with $d(\sigma,\sigma')\leq k$.
Set $c'_{t,i}$, $b_{t,i}$, and $s_{t,i}$ as in
\cref{thm:fpt-nvmax}; by the same argument those variables satisfy
conservation, presence, surplus, and budget.
For each ordered pair $i\neq j$, if $\sigma'_j=\emptyset$, set
$w_{0,i,j}:=1$ and $w_{t,i,j}:=0$ for all $t\in[T]$. Otherwise,
$\sigma'_j$ contains at least one good. By EF1, there exists a witness good
$g^*\in\sigma'_j$ such that $v_i(\sigma'_i)\geq v_i(\sigma'_j)-v_i(g^*)$.
Let $t^*$ be the type of $g^*$, set $w_{t^*,i,j}:=1$, and set
$w_{t,i,j}:=0$ for all $t\in\{0,\dots,T\}\setminus\{t^*\}$. This satisfies
\emph{witness-select}. It also satisfies \emph{witness-valid}, since
$g^*\in\sigma'_j$ implies $c'_{t^*,j}\geq 1$, hence $b_{t^*,j}=1$.
Finally, the \emph{EF1} constraint holds because $v_i(\sigma'_i)
    \geq
    v_i(\sigma'_j)-v_i(g^*)
    =
    v_i(\sigma'_j)-v_i(g_{t^*})$.
All other \emph{EF1} constraints are relaxed by the big-$M$ term. Hence
the extended ILP is feasible.

\smallskip \noindent\textbf{Running time.} The ILP has $p = 2Tn + (T+1)n^2 \leq (T+1)n(n+2) = ((v_{\max}+1)^n+1)n(n+2)$ integer variables in total, a quantity depending only on $n$ and $v_{\max}$. By Lenstra~\cite{Lenstra83}, feasibility of an ILP with $p$ integer variables and input encoding size $L$ can be decided in time $p^{O(p)}\cdot\mathrm{poly}(L)$. The input size is $\mathrm{poly}(m,n,\log v_{\max})$, yielding overall running time $$\bigl(((v_{\max}+1)^n+1)n(n+2)\bigr)^{O(n^2(v_{\max}+1)^n)}\cdot \mathrm{poly}(m,n,\log v_{\max}).$$ 
\end{proof}

\cref{thm:fpt-nvmax} is FPT only in $n+v_{\max}$: the budget~$k$ does not enter the exponent, so it is preferable when $k$ is large or unbounded. Its cost is doubly exponential in $n$, with base $(v_{\max}+1)^n$. \cref{thm:fpt-nvmaxk} instead parameterizes by~$k$, giving running time linear in~$m$ and singly exponential in~$n^2$, unlike \cref{thm:fpt-nvmax,thm:fpt-ef1-nvmax}. Since $(2kv_{\max})^{n^2}\ll (2n(v_{\max}+1)^n)^{n(v_{\max}+1)^n}$ for $k\ll (v_{\max}+1)^n$, the DP is much faster for small reallocation budgets.


\begin{theorem}\label{thm:fpt-nvmaxk}
   For $\beta \in \{EF,EF1,EFX\}$ \realloc{} is solvable in
    $\mathcal{O}\bigl(m \cdot n^3 \cdot (k+1) \cdot (2kv_{\max}+1)^{n^2} \cdot
    (v_{\max}+1)^{n^2}\bigr)$ time.
\end{theorem}

\begin{proof}
We use a dynamic program that processes the goods $g_1,\dots,g_m$ one by one. For a partial assignment of $\{g_1,\dots,g_j\}$, we record:
\begin{itemize}
    \item $\ell \in \{0,\dots,k\}$: the number of reallocations used so far.
    \item the \emph{delta valuation matrix} $\mathbf{d} = (d_{i,i'})_{i,i'\in[n]}$, where
    $d_{i,i'} := v_{i'}(\sigma'_i) - v_{i'}(\sigma_i)$ is the change in the value
    of $a_i$'s bundle as assessed by $a_{i'}$, caused by goods processed so far.
    Since at most $k$ goods contribute to this change and each has value at most $v_{\max}$,
    we have $d_{i,i'} \in [-k v_{\max},\, k v_{\max}]$.
    \item the \emph{max-good matrix} (or resp. min-good) $\mathbf{r} = (r_{i,i'})_{i,i'\in[n]}$, where
    $r_{i,i'} := \max_{g \in \sigma'_i} v_{i'}(g) \in [0, v_{\max}]$ is the highest value, according to agent $a_{i'}$, of any good in agent $a_i$'s current bundle.
    For EFX, we analogously maintain a min-good matrix, initializing every entry to $v_{\max}$ and replacing $\max$ by $\min$ in the updates.
\end{itemize}
The DP table $T[j, \ell, \mathbf{d}, \mathbf{r}] \in \{0,1\}$ records whether the above state is reachable for the first $j$ goods using exactly $\ell$ reallocations.

\smallskip
\noindent There are $m$ layers, $k+1$ values for $\ell$,
$(2kv_{\max}+1)^{n^2}$ choices for $\mathbf{d}$,
and $(v_{\max}+1)^{n^2}$ choices for $\mathbf{r}$.
The total number of states is therefore
\[
    m \cdot (k+1) \cdot (2kv_{\max}+1)^{n^2} \cdot (v_{\max}+1)^{n^2}
    \;=\; m \cdot f(n, v_{\max}, k),
\]
which is \emph{linear in $m$} for fixed $n, v_{\max}, k$.

\smallskip
\noindent Let $a_s$ be the initial owner of good $g_{j+1}$. For each possible recipient $a_u$, define the updated max-good matrix $\mathbf{r}^{(u)}$ by
\[
    r^{(u)}_{a,i'} :=
    \begin{cases}
        \max\{r_{a,i'}, v_{i'}(g_{j+1})\} & \text{if } a=a_u,\\
        r_{a,i'} & \text{otherwise,}
    \end{cases}
    \qquad \forall\, a,i'\in[n].
\]
The recurrence has two cases. In the first case, $g_{j+1}$ stays with its initial agent $a_s$. Then we set \[T[j{+}1,\; \ell,\; \mathbf{d},\; \mathbf{r}^{(s)}] \;\mathrel{\vee}=\; T[j,\; \ell,\; \mathbf{d},\; \mathbf{r}].\]

In the second case, $g_{j+1}$ is reallocated to some agent $a_t$ with $t\neq s$ and $\ell<k$. Define the updated delta matrix $\mathbf{d}^{(s\to t)}$ by
\[
    d^{(s\to t)}_{a,i'} :=
    \begin{cases}
        d_{a,i'} - v_{i'}(g_{j+1}) & \text{if } a=a_s,\\
        d_{a,i'} + v_{i'}(g_{j+1}) & \text{if } a=a_t,\\
        d_{a,i'} & \text{otherwise.}
    \end{cases}
\]
If $\bigl|d^{(s\to t)}_{a,i'}\bigr|\le k\cdot v_{\max}$ for all $a,i'$, we set $ T[j+1,\ell+1,\mathbf{d}^{(s\to t)},\mathbf{r}^{(t)}]
    \;\mathrel{\vee}=\;
    T[j,\ell,\mathbf{d},\mathbf{r}]$.
Otherwise, the transition is skipped.

Each transition updates $n^2$ entries of $\mathbf{d}$ and $\mathbf{r}$. Since there are at most $n$ choices for the recipient in the second case, processing one DP cell takes $\mathcal{O}(n^3)$ time. 

\medskip
\noindent We prove by induction on $j$ that $T[j, \ell, \mathbf{d}, \mathbf{r}] = 1$ if and only if
there exists an assignment of $\{g_1, \dots, g_j\}$ to agents that uses exactly $\ell$
reallocations, 
induces delta matrix $\mathbf{d}$, and induces max-good matrix $\mathbf{r}$.

\smallskip
\noindent\textit{Base case} ($j = 0$).
No good has been processed, so no reallocation has been used and both matrices are zero. Hence we initialise $T[0, 0, \mathbf{0}, \mathbf{0}] = 1$ and set all other entries to $0$.

\smallskip
\noindent\textit{Inductive step} ($j \to j+1$).
Assume the claim holds after processing $g_1,\dots,g_j$. Consider any assignment of $\{g_1,\dots,g_{j+1}\}$ using $\ell'$ reallocations and inducing state $(\mathbf{d}', \mathbf{r}')$.
Let $(\mathbf{d}, \mathbf{r})$ be the state induced by the restriction to $\{g_1,\dots,g_j\}$,
using $\ell$ reallocations.
By the induction hypothesis, $T[j, \ell, \mathbf{d}, \mathbf{r}] = 1$.

\begin{itemize}
    \item If $g_{j+1}$ is \textbf{kept} by its initial agent $a_s$,
    then $\ell' = \ell$, $\mathbf{d}' = \mathbf{d}$, and $r'_{s,i'} = \max(r_{s,i'}, v_{i'}(g_{j+1}))$
    with all other $r'$ entries unchanged.
    The transition rule applies exactly this update, so
    $T[j+1, \ell', \mathbf{d}', \mathbf{r}'] \leftarrow 1$.

    \item If $g_{j+1}$ is \textbf{reallocated} to $a_t$ ($t \neq s$),
    then $\ell' = \ell + 1$, $d'_{s,i'} = d_{s,i'} - v_{i'}(g_{j+1})$,
    $d'_{t,i'} = d_{t,i'} + v_{i'}(g_{j+1})$, and $\mathbf{r}'$ updated accordingly.
    The transition rule for the reallocation performs exactly this update,
    so $T[j+1, \ell', \mathbf{d}', \mathbf{r}'] \leftarrow 1$.
\end{itemize}
Conversely, every entry set to $1$ at layer $j+1$ is obtained by applying one of the two transitions to an entry set to $1$ at layer $j$. By construction, this corresponds to a valid extension of the witnessed partial assignment. This completes the induction.

\smallskip
\noindent\textit{Acceptance.}
At layer $j = m$, the DP has computed all reachable states.
Since $v_{i'}(\sigma'_i) = v_{i'}(\sigma_i) + d_{i,i'}$ and the initial values $v_{i'}(\sigma_i)$ are
fixed input constants, the EF1 condition for a state $(\ell,\mathbf{d},\mathbf{r})$ is
\[
    \forall\, i,i'\in[n]: \quad
    v_i(\sigma_i)+d_{i,i}
    \;\ge\;
    v_i(\sigma_{i'})+d_{i',i}-r_{i',i}.
\]
This condition is checkable directly from $(\mathbf{d},\mathbf{r})$ and the input. We accept if some state $(\ell,\mathbf{d},\mathbf{r})$ with $\ell\le k$ satisfies it.
For EFX, we instead store the corresponding min-good matrix $\mathbf{s}$, where $ s_{i',i}:=\min_{g\in \sigma'_{i'}} v_i(g)$,
with the convention that $s_{i',i}=v_{\max}+1$ if $\sigma'_{i'}=\emptyset$. The EFX condition is then
\[
    \forall\, i,i'\in[n]: \quad
    v_i(\sigma_i)+d_{i,i}
    \;\ge\;
    v_i(\sigma_{i'})+d_{i',i}-s_{i',i}.
\]
This is checked analogously.

\smallskip
\noindent The total running time is $\mathcal{O}\bigl(m \cdot n^3 \cdot f(n, v_{\max}, k)\bigr)$,
where \[f(n, v_{\max}, k) = (k+1)(2kv_{\max}+1)^{n^2}(v_{\max}+1)^{n^2}.\]
This is linear in $m$ and a computable function of $n + v_{\max} + k$, establishing FPT with respect to $n + v_{\max} + k$.
\end{proof}

\section{Conclusion}
We studied \realloc{} for
$\beta \in \{\text{EF}, \text{EF1}, \text{EFX}\}$, providing a
comprehensive complexity map across valuation classes and parameterizations.

A key takeaway is that, under a reallocation budget, the usual algorithmic separation between EF and its relaxations largely disappears. In most of the settings we study, EF1, EFX, and EF exhibit the same hardness behavior, and our algorithmic approaches extend across the three notions with only minor modifications.

Open directions include extending the study to other fairness notions, such as proportionality and the maximin share, and to mixed manna (goods and chores). Such extensions would yield a more complete picture.
Finally, on the optimisation side, designing approximation algorithms for the minimum reallocation distance---the smallest $k$ for which a $\beta$-fair allocation is reachable---and investigating structurally restricted reallocation models (e.g.\ pairwise exchanges or gift-giving chains) are appealing directions for future work.
\bibliographystyle{splncs04}
\bibliography{main}

\newpage

\end{document}